\documentclass[10pt,final,doublecolumn]{IEEEtran}
\usepackage{subfigure}
\usepackage{amsmath}
\usepackage{latexsym}
\usepackage{graphicx}
\usepackage{bbding}
\usepackage{indentfirst}
\usepackage{setspace}
\usepackage{cite,balance}
\usepackage{float}
\usepackage{epstopdf}
\usepackage{amssymb}
\usepackage{amsfonts}
\usepackage{enumerate}
\usepackage{multicol}
\usepackage{color}
\usepackage{stfloats}
\usepackage{bm}
\usepackage{amssymb}
\usepackage{stfloats}
\usepackage{algorithm,algpseudocode,float}
\usepackage{lipsum}
\newtheorem{theorem}{\bf \emph{\underline{Theorem}}}

\newtheorem{lemma}{\bf \emph{\underline{Lemma}}}

\newtheorem{example}{\bf \emph{\underline{Example}}}
\newtheorem{remark}{\bf \emph{\underline{Remark}}}

\def\l{\left}
\def\r{\right}
\def\({\left(}
\def\){\right)}

\def\b0{{\mathbf{0}}}

\newcommand{\tr}{\mathrm{tr}}
\newcommand{\diag}{\mathrm{diag}}

\newcommand{\nn}{\nonumber}

\IEEEoverridecommandlockouts
\allowdisplaybreaks[4]

\begin{document}
\title{Target Sensing with Intelligent Reflecting Surface: Architecture and Performance}
\author{{Xiaodan Shao,~\IEEEmembership{Member,~IEEE}, Changsheng You, \IEEEmembership{Member, IEEE},  Wenyan Ma, \IEEEmembership{Student Member, IEEE}, Xiaoming Chen, \IEEEmembership{Senior Member, IEEE},  and  Rui Zhang, \IEEEmembership{Fellow, IEEE}}
	 \vspace{-18pt}
\thanks{The work of X. Chen is supported by the Natural Science Foundation of China under grant U21A20443. The work of R. Zhang was supported by Ministry of Education, Singapore under Award T2EP50120-0024 and by Advanced Research and Technology Innovation Centre (ARTIC) of National University of Singapore under Research Grant R-261-518-005-720. \emph{(Corresponding author: Changsheng You)}.}
	 \vspace{-18pt}
	\thanks{X. Shao is with the College of Information Science and Electronic Engineering, Zhejiang University, Hangzhou 310016, China, and
he is also a visiting scholar at National University of Singapore, Singapore 117583, (Email: shaoxiaodan@zju.edu.cn).}

\thanks{C. You was with the Department of Electrical and Computer Engineering, National University of Singapore, Singapore 117583, he is now with the Department of Electrical and Electronic Engineering, Southern University of Science and Technology (SUSTech), Shenzhen 518055, China, (Email: {youcs@sustech.edu.cn}).}

\thanks{X. Chen is with the College of Information Science and Electronic Engineering, Zhejiang University, Hangzhou 310016, China, (Email: chen\_xiaoming@zju.edu.cn).}

\thanks{W. Ma is with the Department of Electrical and Computer Engineering, National University of Singapore, Singapore 117583, (Email: e0787961@u.nus.edu).}

\thanks{R. Zhang is with the Department of Electrical and Computer Engineering, National University of Singapore, Singapore 117583, (Email: elezhang@nus.edu.sg).}
}
\maketitle
\begin{abstract}
Intelligent reflecting surface (IRS) has emerged as a promising technology to reconfigure the radio propagation environment by dynamically controlling wireless signal's amplitude and/or phase via a large number of reflecting elements. In contrast to the vast literature on studying IRS's performance gains in wireless communications, we study in this paper a new application of IRS for sensing/localizing targets in wireless networks. Specifically, we propose a new \emph{self-sensing IRS} architecture where the IRS controller is capable of transmitting probing signals that are not only directly reflected by the target (referred to as the direct echo link), but also consecutively reflected by the IRS and then the target (referred to as the IRS-reflected echo link). Moreover, dedicated sensors are installed at the IRS for receiving both the direct and IRS-reflected echo signals from the target, such that the IRS can sense the direction of its nearby target by applying a customized multiple signal classification (MUSIC) algorithm. However, since the angle
estimation mean square error (MSE) by the MUSIC algorithm is intractable, we propose to
optimize the IRS passive reflection for maximizing the average echo signals' total power at the IRS sensors and derive the resultant
Cramer-Rao bound (CRB) of the angle estimation MSE. Last, numerical results are presented to show the effectiveness of the proposed new IRS sensing architecture and algorithm, as compared to other benchmark sensing systems/algorithms.

%The target localization performance of existing radar systems is usually constrained by the

\end{abstract}

%\newpage
\begin{IEEEkeywords}
Intelligent reflecting surface (IRS), passive reflection, wireless sensing and localization.
\end{IEEEkeywords}

\IEEEpeerreviewmaketitle

\section{Introduction}
The future sixth-generation (6G) wireless systems need to enable the emerging location-aware applications such as virtual reality, robot navigation, autonomous driving, and so on. These applications impose more stringent requirements on both the communication and sensing performance of today's wireless systems, such as ultra-high data rate, ubiquitous and seamless coverage, extremely-high reliability and ultra-low latency, as well as high-precision/resolution sensing \cite{6G1,6G2}.
With the development of massive multi-input multi-output (MIMO) and millimeter wave (mmWave) communication technologies, the base station (BS) in communication systems has the capability to achieve high-resolution sensing in the angular domain. This thus gives rise to the emerging research area of joint communication and sensing, which advocates to share the hardware, platform, and radio resource in the design and use of communication and sensing systems to achieve their substantially enhanced performance with reduced cost  \cite{cui2021integrating,liu2020joint,ma2020joint,feng2020joint}.

%mono-static BS sensing
%bi-static mobile sensing

Specifically, in the conventional mono-static BS sensing system,
the transmit and receive antennas are co-located at the BS or the BS antenna array is exploited for both transmitting and receiving radar probing signals in a full-duplex (FD) manner \cite{mono}.
%the BS antenna array is separated into two sets or exploited as a whole to transmit and receive radar probing signals.
%s such, the mono-static sensing system
%the BS antenna array is divided into two sets for transmitting and receiving radar probing signals, respectively,
%the BS antenna array is exploited for both transmitting and receiving radar probing signals. Such a so-called mono-static radar sensing system
This mono-static sensing system generally entails a small number of  angle-and-distance parameters to be estimated, due to the same direction-of-angle (DOA) and distance over the BS$\to$target and target$\to$BS links.
However, the BS receive array may suffer non-negligible interference from its transmit array. Moreover, the localization performance of the BS is deteriorated with the increase of the target distance, due to the severe product-distance round-trip path-loss as well as random blockages between the BS and targets \cite{liangradar}.
Alternatively, in the bi-static mobile sensing system, the transmit and receive arrays are placed in different sites (e.g., BSs and/or mobile devices). For example, a  BS transmits probing signals, while a multi-antenna mobile device receives the echo signals reflected by the target for its localization. As such, the bi-static sensing system has the potential to improve the sensing performance over its mono-static sensing counterpart, due to the smaller product-distance path-loss and less susceptibility to interference, provided that the mobile device is near the target and far away from the BS. Nevertheless, the bi-static sensing system generally requires more parameters to be estimated for target localization, due to the different DOAs and distances over the BS$\to$target and target$\to$device links.
\begin{figure*}[t]
  \centering
\includegraphics [width=180mm] {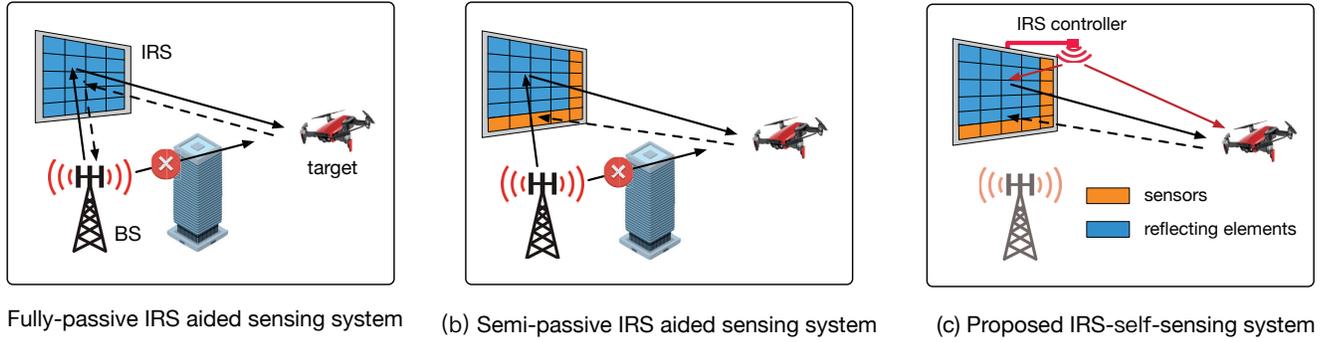}
\caption{Different IRS-aided sensing systems.}
\label{syscomp}
\end{figure*}

To meet the high demands of future 6G systems in terms of both communication and sensing, the concept of smart radio environment has been recently proposed \cite{liaskos2018new}, which essentially leverages the digitally-controlled metasurface, termed \emph{intelligent reflecting surface} (IRS) (or reconfigurable intelligent surface (RIS) equivalently), to reconfigure the wireless propagation environment in favor of both wireless communication and radar sensing \cite{wu2021intelligent,basar19_survey,irs1,irs2,irs3}. In particular,
by dynamically controlling phase shifts via a large number low-cost reflecting elements, IRS is able to realize many desired functions, such as bypassing environment obstacles, reshaping channel realizations/distributions, and increasing the multi-antenna/MIMO channel rank.
These functions have been exploited and thoroughly investigated in the existing literature for improving the communication performance in various wireless systems, by developing efficient designs of IRS passive beamforming/reflection, channel estimation, and placement (see, e.g., \cite{JR:wu2018IRS,huang2018largeRIS,yu2020robust,you2019progressive,zhang2021intelligent,you2020deploy,you2021enabling}).
On the other hand, in terms of radar sensing, IRS is able to establish a line-of-sight (LoS) link with the target in its vicinity, which is particularly helpful when the direct sensing link between the BS and target is blocked. Moreover, in addition to the direct BS-target link, IRS provides an extra LoS reflected link to sense the target from a different angle, thus potentially enhancing the sensing performance.
The above benefits have motivated active research recently on designing  different IRS systems/algorithms for improving the sensing performance.
%, which remains largely uncharted.

Among others, a fully-passive IRS aided sensing system was studied in \cite{sen1,sen2,sen3,sen4,sen5,sen6}, where an IRS without any transmit/receive radio-frequency (RF) chains is employed to help estimate the direction of a nearby target in the challenging scenario where the direct link between the BS and target is blocked. In this case, the BS sends probing signals, which are reflected by the IRS to illuminate the target with dynamically tuned beam directions. After the IRS-reflected beam hits the target, the echo signal is consecutively reflected by the target and IRS (again), and finally received at the BS, which estimates the direction of the target with respect to (w.r.t.) the IRS. However, this approach may suffer severe path-loss over the multiple (three-hop) signal reflections, i.e., BS$\to$IRS$\to$target$\to$IRS$\to$BS as shown in Fig.~\ref{syscomp}(a), and thus degraded  DOA estimation performance. To tackle this issue, an alternative IRS-aided sensing method is to utilize the semi-passive IRS, where the IRS is equipped with dedicated (low-cost) sensors to receive signals for facilitating channel estimation in communications as well as target localization \cite{semi1,semi2,semi3,wu2021intelligent}. Specifically, one efficient way is letting the IRS reflect the probing signals sent by the BS to the target, while the IRS sensors receive the echo signals reflected by the target for estimating its angle w.r.t. the IRS.
Various approaches can be applied to estimate the target angle w.r.t. the IRS, such as Capon's minimum variance method \cite{cano}, maximum likelihood estimation (MLE) \cite{mle}, and the
subspace-based algorithms (e.g., multiple signal classification (MUSIC) algorithm \cite{musica} and estimation of signal parameters via rotational invariance techniques (ESPRIT) algorithm \cite{esp}).
Nevertheless, the angle estimation performance is still constrained by the sensors' received signal-to-noise ratio (SNR), due to the high path-loss of the two-hop signal reflections, i.e., BS$\to$IRS$\to$target$\to$sensors,
as illustrated in Fig.~\ref{syscomp}(b).

To improve the sensing performance of the aforementioned IRS aided sensing systems, we propose in this paper a new \emph{IRS-self-sensing} system as shown in Fig.~\ref{syscomp}(c), where the IRS reflecting elements are illuminated by the probing signals from the IRS controller\footnote{IRS controller is attached to each IRS for controlling its signal reflection as well as communicating with associated BS/mobile devices for exchanging control signals in IRS-aided communications; thus, the IRS controller needs to possess both transmit and receive RF modules and it can also send probing signals for target sensing as considered in this paper \cite{zhengIRS, qingtoward}.} in short distance and form reflected beams to sense the locations of its nearby targets. Moreover, receiving sensors are installed at the IRS to receive the echo signals from two types of links: 1) the two-hop IRS-reflected echo link,
% one is consecutively reflected over the IRS reflecting link,
  i.e., IRS controller$\to$IRS elements$\to$target$\to$IRS sensors; and 2) the one-hop direct echo link,
%   the other is directly reflected from the target over the direct link,
   i.e., IRS controller$\to$target$\to$IRS sensors. This thus greatly reduces the product-distance path-loss in the IRS-reflected echo link as compared to both the conventional fully- and semi-passive IRS aided sensing systems (see Figs. 1(a) and 1(b)), since the probing signals are transmitted from the IRS controller (instead of the BS) that is placed very close to the IRS reflecting elements in practice. Moreover, the MUSIC algorithm can be applied to estimate the target DOA with better resolution based on the received echo signals from both links at the IRS sensors. On the other hand, from a network implementation perspective, the self-sensing IRS does not involve the BS or any mobile device in the sensing process, hence greatly reducing the communication and control overhead required for target sensing. Thus, it is beneficial to deploy such low-cost self-sensing IRSs in the network to locate their nearby targets in a distributed manner, which can significantly alleviate the sensing task at the BS/device side in wireless networks. In the following, we summarize the main contributions of this paper.

First,
we apply a customized MUSIC algorithm for our considered IRS-self-sensing system to estimate the DOA of a target near the IRS, based on the received echo signals over both the direct and IRS-reflected echo links. Although the DOA estimation MSE by the MUSIC algorithm is intractable, we formulate an optimization problem to optimize the IRS passive reflection for maximizing the average echo signals' total power at IRS sensors. We show that in the presence of both the direct and IRS-reflected echo links, the discrete Fourier transform (DFT) based IRS passive reflection is optimal, which generates an omnidirectional beampattern in the angular domain for target sensing. Second, we show that the received signal power due to the IRS-reflected echo link linearly increases with the number of IRS reflecting elements, and it is larger than that due to the direct echo link  when the number of IRS reflecting elements is sufficiently large. Moreover, we analytically show that it is beneficial to employ the IRS controller for sending probing signals as compared to the benchmark scheme by using instead a nearby mobile device. Furthermore, the Cramer-Rao bound (CRB) of the target DOA estimation MSE is analytically derived. Last, numerical results are presented to validate the performance gain of the proposed IRS-self-sensing system as compared to various benchmark  sensing systems/algorithms.

The rest of this paper is organized as follows.  We first present in Section II the proposed new IRS-self-sensing system and its sensing protocol. In Section III, we present the customized DOA estimation algorithm and the optimal IRS passive reflection for target sensing. Subsequently, we analytical show the advantages of the proposed new IRS architecture and the resultant CRB for the DOA estimation in Section IV. Numerical results are provided in Section V, followed by the conclusions in Section VI.

\emph{Notations}: $\left \|\cdot\right \|_2$ denotes the Euclidean norm, $\text{vec}(\cdot)$ denotes the operation that stacks the columns of a matrix, $\mathrm{diag}({\bf x})$ denotes a diagonal matrix with the diagonal entries specified by vector ${\bf x}$. $\Re(x)$ and $\Im(x)$ denote the real and imaginary parts of a complex number $x$. For a complex symbol, $(\cdot)^{\dagger}$, $(\cdot)^H$, and $(\cdot)^T$ denote its complex conjugate,  conjugate transpose, and transpose, respectively. For a square matrix $\mathbf{S}$, $\mathbf{S}\succeq 0$ means that $\mathbf{S}$ is positive semi-definite, and $\tr(\mathbf{S})$ denotes its trace. The Hadamard product is denoted by $\odot$, and  the Kronecker product is denoted by $\otimes$. The distribution of a circularly symmetric complex Gaussian (CSCG) random variable with mean $u$ and covariance $\sigma^2$ is denoted by $\mathcal{CN}(u,\sigma^2)$, $\backsim$ stands for  ``distributed as", $\mathbb{E}(\cdot)$ denotes the statistical expectation, and $\dot{\mathbf{a}}(x)$ denotes the first-order partial derivative w.r.t. $x$.
%    $\mathbf{x}\in \mathbb{C}^{n}$ is said to follow a vector-valued Gaussian  distribution with mean  of the form $\mathcal{CN}(\mathbf{x}|\mathbf{u},\boldsymbol{\Sigma})$.

\section{System Model}
\begin{figure}[t]
\setlength{\abovecaptionskip}{-0.cm}
\setlength{\belowcaptionskip}{0.cm}
  \centering
\includegraphics [width=80mm] {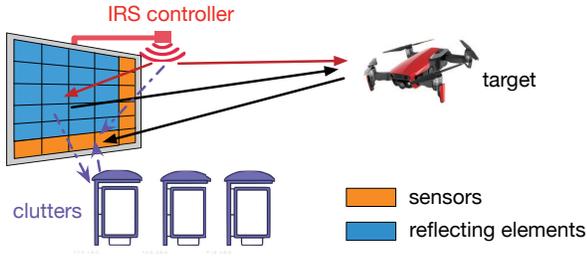}
\caption{An IRS-self-sensing system.}
\label{sysmodel}
\end{figure}

We consider the proposed IRS-self-sensing  system as shown in Fig.~\ref{sysmodel}, where an IRS is deployed in the network to sense the location of an unknown target in the presence of $L$ clutters under the assumption that the IRS location is known \emph{a priori}\footnote{The proposed algorithm for the single-target localization can be readily extended to the multi-target case by using the MUSIC algorithm if the targets are sufficiently far-apart, while more advanced techniques (e.g., spatial smoothing \cite{smooth}) need to be developed to resolve the targets' DOAs if they are closely located in the angular domain.}. The IRS consists of $N_0=N_{\rm h}\times N_{\rm v}$ reflecting elements, where $N_{\rm h}$ and $N_{\rm v}$ denote respectively the number of horizontal and vertical elements with the inter-element (horizontal and vertical) distance denoted by $d_{\rm I}$.
%, while the results of this paper can be extended to  the case of the uniform planar array (UPA) IRS by considering its three-dimensional passive beamforming and localization.
The reflecting elements are illuminated by the probing signals from a single omnidirectional-antenna IRS controller that is capable of transmitting signals as well as dynamically controlling the IRS passive reflection for localizing the target. Moreover, to estimate both the azimuth and elevation angles of the target, $M_0$ (assumed to be an even number) sensors are installed adjacent to the IRS reflecting elements to receive the echoed signals by the target as shown in Fig.~\ref{sysmodel}, where $M_{\rm h}=M_{\rm v}=(M_0+1)/2$ sensors are placed along the horizontal and vertical axes, respectively, with the inter-sensor distance given by $d_{\rm s}$. For simplicity, we focus on the target's azimuth angle estimation  in this paper based on the received signals on the $M_{\rm h}$ horizontal sensors, while the results can be extended to estimate its elevation angle as well.
\subsection{Radar Channel Model}
We assume a narrow-band sensing system, where the IRS controller consecutively sends probing signals over $T$ snapshots, within which all the associated channels are assumed to remain static. At each snapshot, the sensors receive the signals transmitted from the IRS controller and that reflected by the target as well as clutters, with and without the IRS reflection. Specifically, the received signals reflected by the target undergo two types of links: 1) the two-hop IRS-reflected echo link, i.e., IRS controller$\to$IRS elements$\to$target$\to$IRS sensors; and 2) the single-hop direct echo  link, i.e., IRS controller$\to$target$\to$IRS sensors. These two links are modeled as follows in detail, respectively.

First, consider the IRS-reflected echo link associated with the target. Let ${\bf u}(\bar{\vartheta}, \bar{N})$ denote the steering vector function, which is defined as
\begin{align}\label{array}
{\bf u}(\bar{\vartheta},  \bar{N})\triangleq [e^{\frac{-\jmath (\bar{N}-1)\pi\bar{\vartheta}}{2}},  e^{\frac{-\jmath (\bar{N}-3)\pi\bar{\vartheta}}{2}}, \dots, e^{\frac{\jmath (\bar{N}-1)\pi\bar{\vartheta}}{2}}]^T,
\end{align}
where $\bar{N}$ (assumed to be an even number) denotes the uniform linear array (ULA) size and $\bar{\vartheta}$ denotes the constant phase difference between the observations at two adjacent elements/sensors. Then, the IRS controller$\to$IRS elements channel, denoted by
%As the IRS controller is placed very close to the reflecting elements, the controller$\to$elements channel, denoted by
${\bf h}_{{\rm CI}}\in\mathbb{C}^{N_0\times 1}$,  can be modeled as follows based on the far-field LoS channel model.\footnote{We assume the element-wise channel is in the far field from the IRS controller to IRS elements, while the results can be extended to the case of near-field channel condition when the IRS controller is placed extremely close to IRS elements \cite{near1,near2}.}
% can be modeled based on the near-field line-of-sight (LoS) channel model. Specifically, let $d_{{\rm CI}, n}$ denote the distance between the IRS controller and the $n$-th IRS reflecting element with $n\in\mathcal{N}\triangleq\{1,\cdots, N\}$. Then the $n$-th controller$\to$elements channel coefficient can be modeled as
\begin{equation}
{\bf h}_{{\rm CI}}=
% \frac{\sqrt{\beta_0} }{d_{{\rm CI}}}
%\sqrt{\alpha_{\rm CI}} e^{\frac{\jmath 2\pi d_{\rm CI}} {\lambda}}
{\alpha}_{{\rm CI}} {\bf a}(\theta_{\rm CI, h}, \theta_{\rm CI, v}),
% {\bf u}(\theta_{\rm CI, h}, N_{\rm h}) \otimes {\bf u}(\theta_{\rm CI, v}, N_{\rm v}),
 \label{Eq:CI model}
\end{equation}
where ${\alpha}_{{\rm CI}}=\sqrt{\frac{\lambda^2}{16\pi^2 d_{\rm CI}^2}}e^{\frac{\jmath 2\pi d_{\rm CI}} {\lambda}}$ denotes the complex-valued path gain of the IRS controller$\to$IRS elements channel, $\lambda$ denotes the carrier wavelength, $d_{{\rm CI}}$ denotes the distance between the IRS controller and IRS central element,
%$\alpha_{\rm CI}=\frac{\lambda^2}{16\pi^2 d_{\rm CI}^2 }$ represents the controller$\to$elements channel power gain with $\lambda$ denoting the carrier wavelength,
and $\theta_{\rm CI, h}\in[-\pi/2,\pi/2]$ and $\theta_{\rm CI, v}\in[-\pi/2,\pi/2]$ denote respectively the \emph{physical} azimuth and vertical angles-of-arrival (AoA) at the IRS. Therein, ${\bf a}(\cdot)$ denotes the response vector function of the IRS, which is defined as
\begin{align}
 {\bf a}(\theta_{\rm CI, h}, \theta_{\rm CI, v})= {\bf u}(\phi_{\rm CI, h}, N_{\rm h}) \otimes {\bf u}(\phi_{\rm CI, v}, N_{\rm v}),
\end{align}
where $\phi_{\rm CI, h}\triangleq\frac{2d_{\rm I}}{\lambda}\sin(\theta_{\rm CI, h})\sin(\theta_{\rm CI, v})\in[-\frac{2d_{\rm I}}{\lambda},\frac{2d_{\rm I}}{\lambda}]$ and $\phi_{\rm CI, v}\triangleq\frac{2d_{\rm I}}{\lambda}\cos(\theta_{\rm CI, v})\in[-\frac{2d_{\rm I}}{\lambda},\frac{2d_{\rm I}}{\lambda}]$ are referred to as the horizontal and vertical spatial directions, respectively. Under the far-field condition, the AoAs from the target to IRS sensors can be assumed to be the same as the angles-of-departure (AoDs) from the IRS elements to target. As such, the echo channel of the  IRS elements$\to$target$\to$IRS (horizontal) sensors link can be modeled as
% matrix (or the elements $\to$ target $\to$ sensors channel) for a target in the direction $(\theta_{\rm IT, h}, \theta_{\rm IT, h})$
%is given by
\vspace{-0.39em}
\begin{align}
{\bf H}_{{\rm ITS}} = \alpha_{\rm r}{\bf b}(\theta_{\rm IT, h}){\bf a}^T(\theta_{\rm IT, h}, \theta_{\rm IT, v}),
\end{align}
where $\theta_{\rm IT, h}$ and $\theta_{\rm IT, v}$ denote the azimuth and vertical AoDs from the IRS to target, respectively, ${\bf b}(\theta_{\rm IT, h})={\bf u}(\frac{2d_{\rm s}}{\lambda}\sin(\theta_{\rm IT, h}), M_{\rm h})\in\mathbb{C}^{M_{\rm h}\times 1}$ denotes the receive response vector from the target to IRS sensors, and
$\alpha_{\rm r}= \beta_{\rm r}  G_{\rm r} $ represents the complex-valued  path gain with $\beta_{\rm r}\backsim \mathcal{CN}(0, 1)$  denoting the small-scale complex channel gain and  $G_{\rm r}=\sqrt{\frac{{{\lambda ^2}\kappa }}{{64{\pi ^3}{d_{\rm IT}^4}}}}$ denoting the signal attenuation caused by the propagation from IRS to the target and then from the target to IRS sensors as well as the scattering process. Therein, $d_{\rm IT}$ denotes the distance between the IRS and target, and  $\kappa$ denotes the radar cross section (RCS), which is a measurement of power scattered in a given direction when a target is illuminated by an incident wave.
%Based on the above,
%and . On the other hand, we assume that the target locates sufficiently far from the IRS and there exist an LoS path between the IRS and target.
% Under the far-field propagation model and given the (physical) azimuth angle of the target with respect to the IRS, denoted by $\theta\in[-\pi/2,\pi/2]$, the receive response vector from the IRS reflecting elements to the target can be expressed as
%\begin{align}
%{\bf a}^T(\theta) = {\bf u}^T(\phi, N),
%\end{align}
%where $\phi\triangleq\frac{2d_{\rm I}}{\lambda}\cos(\phi)\in[-\frac{2d_{\rm I}}{\lambda},\frac{2d_{\rm I}}{\lambda}]$ is referred to as the  \emph{spatial} (azimuth) angle.
%As the target is at the far-filed of the IRS, we assume the plane wave reflected at the sensors.
%Similarly, the receive response vector from the target to the IRS sensors, denoted by ${\bf b}\in\mathbb{C}^{M\times1}$, is given by
%%\footnote{Note that the receive response vector of the horizontal sensors is affected by the azimuth angle only.}
%\begin{align}
%{\bf b}(\theta) = {\bf u}(\phi, M).
%\end{align}
Let ${\boldsymbol{\varphi}}_0[t]\in\mathbb{C}^{N_0\times 1}$ denote the IRS reflection vector at each snapshot $ t \in \mathcal{T}\triangleq\{1,\cdots, T\}$ with each coefficient $[{\boldsymbol{\varphi}}_0[t]]_n=e^{\jmath \omega_{n}}, n\in\{1, \cdots, N_0\}$ denoting the phase shift of IRS reflecting element $n$. Based on the above, the channel of the echo signal reflected from the target over the IRS-reflected link (i.e., IRS controller$\to$IRS elements$\to$target$\to$IRS sensors) at snapshot $t$ is given by
\begin{align}
&{\bf g}_{\rm r}[t]= {\bf H}_{{\rm ITS}} \diag({\boldsymbol{\varphi}}_0[t]) {\bf h}_{{\rm CI}}\nn\\
&=\alpha_{\rm r}\alpha_{\rm CI} {\bf b}(\theta_{\rm IT, h}) {\bf a}^T(\theta_{\rm IT, h}, \theta_{\rm IT, v}) \diag({\boldsymbol{\varphi}}_0[t])  {\bf a}(\theta_{\rm CI, h}, \theta_{\rm CI, v}).\label{Eq:ref}
\end{align}
%where $\alpha_{\rm r}={\gamma}_{{\rm CI}} {\gamma}_{{\rm ITS}}$.
%  represents the complex-valued reflection path gain with $\gamma_{\rm r}$ denoting the complex-valued {\color{blue}?} and  $G_{\rm r}=\frac{{{\lambda ^2}\kappa }}{{64{\pi ^3}{d_{\rm IT}^4}}}$ denoting the signal attenuation over the product-distance reflection path-loss and the scattering process. Therein, $d_{\rm IT}$ denotes the distance between the IRS and target and  $\kappa$ denotes the radar cross section (RCS),
Next,
%Following the similar far-filed propagation model as for the reflecting link,
 the direct echo link (i.e., IRS controller$\to$target$\to$IRS sensors) can be modeled as
 \vspace{-0.55em}
\begin{align}
{\bf g}_{\rm d}&=\alpha_{\rm d}{\bf b}(\theta_{\rm IT, h}),\label{Eq:dir}
\end{align}
where $\alpha_{\rm d}=\beta_{\rm d}  G_{\rm d}$ denotes the complex-valued reflection path gain with $\beta_{\rm d}\backsim \mathcal{CN}(0, 1)$ being the small-scale complex channel gain and $G_{\rm d}=\sqrt{\frac{{{\lambda ^2}\kappa }}{{64{\pi ^3}{d_{\rm CT}^2d_{\rm IT}^2}}}}$. Herein, $d_{\rm IT}$ denotes the distance between the IRS and target, and $d_{\rm CT}$ denotes the distance from the IRS controller to target.
Note that ${\bf g}_{\rm d}$ remains static over $T$ snapshots.

On the other hand, for the clutters, let ${\bf g}_{{\rm r},I_{\ell}}[t]\in\mathbb{C}^{M_{\rm h}\times 1}$ and ${\bf g}_{{\rm d},I_{\ell}}\in\mathbb{C}^{M_{\rm h}\times 1}$ denote respectively the IRS-reflected and direct echo links associated with each clutter ${\ell}\in 1, \cdots, L$ at each snapshot $t$, which can be similarly modeled as \eqref{Eq:ref} and \eqref{Eq:dir} for the target. In addition, we denote by  ${\bf h}_{{\rm CS}}\in\mathbb{C}^{M_{\rm h}\times 1}$  the IRS controller$\to$IRS sensors channel, which can be modeled similar to the IRS controller$\to$IRS elements channel in \eqref{Eq:CI model}.
As such, the received signals at the IRS (horizontal) sensors at each snapshot $t$ is given by
\begin{align}
{\bf y}_0[t]& = \left({\bf g}_{\rm r}[t] + {\bf g}_{\rm d} +\underbrace{\sum_{\ell=1}^L ({\bf g}_{{\rm r},I_{\ell}}[t] + {\bf g}_{{\rm d},I_{\ell}}) + {\bf h}_{{\rm CS}}}_{\rm background~channel} \right) x[t]\nn\\
&+ {\bf z}_0[t], t\in\mathcal{T},\label{Eq:signal}
\end{align}
where $x[t]$ is the signal sent by the IRS controller at snapshot $t$ which is simply set as $x[t]=1, \forall t$, and ${\bf z}_0[t] \in\mathbb{C}^{M_{\rm h}\times 1}\backsim \mathcal{CN}(\mathbf{0}_{M_{\rm h}}, \sigma_0^2\mathbf{I}_{M_{\rm h}}) $ denotes the zero-mean additive white Gaussian noise (AWGN) vector at snapshot $t$ with the normalized noise power $\sigma_0^2$.

\subsection{Proposed Protocol for IRS-enabled Target Angle Estimation}
The signal model in \eqref{Eq:signal} is assumed to be obtained in a particular range-Doppler bin of interest, for which the range and the Doppler parameters are omitted in the model. Thus, we focus on the (azimuth) angle estimation in this paper\footnote{The explicit target location can be resolved based on the estimated angle and distance, where the IRS-target distance can be estimated by measuring its round-trip running time based on e.g., matched filter \cite{fish, bin}.} The main procedures consist of two phases, namely, the \emph{offline} training phase that estimates the static background channel at the IRS sensors w.r.t. each IRS reflection pattern available, followed by the \emph{online} estimation phase that tunes IRS reflections for estimating the target angle w.r.t. the IRS. Specifically, in the offline training, we assume that the environmental state remains unchanged over a much larger timescale than that of the real-time sensing and
there is no target in the environment; while the IRS applies all reflection patterns/vectors in a predefined IRS reflection matrix (to be specified  in Section~III). As such, when the IRS applies the passive reflection vector ${\boldsymbol{\varphi}}$, the received signal vector at the IRS sensors is given by
\begin{align}
{{\bf y}}_{\rm int}({\boldsymbol{\varphi}}) = \left(\underbrace{\sum_{\ell=1}^L ({\bf g}_{{\rm r},I_{\ell}}({\boldsymbol{\varphi}}) + {\bf g}_{{\rm d},I_{\ell}}) + {\bf h}_{{\rm CS}}}_{{\bf g}_{\rm int}(\boldsymbol{\varphi}):~ \rm background ~ channel } \right) x+ \tilde{\bf z},\label{Eq:signal2}
\end{align}
where ${\bf g}_{{\rm r},I_{\ell}}({\boldsymbol{\varphi}}) $ depends on  the IRS reflection vector $\boldsymbol{\varphi}$ in general. Thus, with $x=1$, the static background  channel given IRS reflection vector $\boldsymbol{\varphi}$ can be estimated as $\hat{{\bf g}}_{\rm int}(\boldsymbol{\varphi})={{\bf y}}_{\rm int}({\boldsymbol{\varphi}})$, which includes the channels from the IRS controller and all nearby clutters to the IRS sensors. Next, consider the online estimation phase. For each IRS reflection vector ${\boldsymbol{\varphi}}_0[t]$ in snapshot $t$, undesired  interference from the IRS controller and all clutters is first removed in  the received signal ${\bf y}_0[t]$ based on the correspondingly estimated background channel, which yields
\begin{align}\label{Eq:yt}
{\bf y}[t]&= {\bf y}_0[t]-\hat{{\bf g}}_{\rm int}(\boldsymbol{\varphi}[t])=({\bf g}_{\rm r}[t] + {\bf g}_{\rm d})x[t]+{\bf z}[t],\nonumber\\
&=\alpha_{\rm r} \alpha_{\rm CI} {\bf b}(\theta_{\rm IT, h}) {\bf a}^T(\theta_{\rm IT, h}, \theta_{\rm IT, v}) \diag({\boldsymbol{\varphi}}_0[t])  {\bf a}(\theta_{\rm CI, h}, \theta_{\rm CI, v})\nn\\
&+\alpha_{\rm d}{\bf b}(\theta_{\rm IT, h})+{\bf z}[t],
\end{align}
where $\mathbf{z}[t]=(\mathbf{z}_0[t]-\tilde{\mathbf{z}})\sim \mathcal{CN}(\mathbf{0},\sigma^2\mathbf{I})$ with $\sigma^2\triangleq2\sigma_0^2$. To facilitate the  IRS passive reflection design over space and time, we denote  ${\boldsymbol{\varphi}}_0[t]$ as the Hadamard product of two vectors, i.e.,  ${\boldsymbol{\varphi}}_0[t]={\boldsymbol{\varphi}}_{\rm h}[t]\otimes{\boldsymbol{\varphi}}_{\rm v}[t]$, where ${\boldsymbol{\varphi}}_{\rm h}=[e^{\jmath \omega_{1}}, \cdots, e^{\jmath \omega_{N_{\rm h}}}]^T$  and  ${\boldsymbol{\varphi}}_{\rm v}=[e^{\jmath \omega_{1}}, \cdots, e^{\jmath \omega_{N_{\rm v}}}]^T$ are referred to as the horizontal and vertical IRS reflection vectors, respectively. Then, we have
\begin{align}
{s}({\boldsymbol{\varphi}}_0[t])&\triangleq{\bf a}^T(\theta_{\rm IT, h}, \theta_{\rm IT, v}) \diag({\boldsymbol{\varphi}}_0[t])  {\bf a}(\theta_{\rm CI, h}, \theta_{\rm CI, v})\nn\\
&=({\bf a}^T(\theta_{\rm IT, h}, \theta_{\rm IT, v})\odot {\bf a}^T(\theta_{\rm CI, h}, \theta_{\rm CI, v})){\boldsymbol{\varphi}}_0[t]\nn\\
&=({\bf u}^T(\phi_{\rm IT, h}, N_{\rm h}) \otimes {\bf u}^T(\phi_{\rm IT, v}, N_{\rm v}))\odot  ({\bf u}^T(\phi_{\rm CI, h}, N_{\rm h}) \nn\\
&~~~\otimes {\bf u}^T(\phi_{\rm CI, v}, N_{\rm v})){\boldsymbol{\varphi}}_0[t]\nn\\
&= [{\bf u}^T(\tilde{\phi}_{\rm IT, h}, N_{\rm h})\otimes {\bf u}^T(\tilde{\phi}_{\rm IT, v}, N_{\rm v})][{\boldsymbol{\varphi}}_{\rm h}[t]\otimes{\boldsymbol{\varphi}}_{\rm v}[t]]\nn\\
&=({\bf u}^T(\tilde{\phi}_{\rm IT, h}, N_{\rm h}){\boldsymbol{\varphi}}_{\rm h}[t])\otimes({\bf u}^T(\tilde{\phi}_{\rm IT, v}, N_{\rm v}){\boldsymbol{\varphi}}_{\rm v}[t]),
%\triangleq{\bf q}^T  {\boldsymbol{\varphi}}_{\rm h}[t] \eta_{\rm r},
\end{align}
where $\tilde{\phi}_{\rm IT, h}\triangleq \phi_{\rm IT, h}+\phi_{\rm CI, h}$ and $\tilde{\phi}_{\rm IT, v}\triangleq \phi_{\rm IT, v}+\phi_{\rm CI, v}$.  For simplicity, we assume that the IRS vertical reflection vector has been aligned. Thus, \eqref{Eq:yt} can be simplified as \footnote{Note that the DOA estimation under the ULA-based IRS can be easily extended to the UPA-based IRS. For example, the IRS controller can first fix its vertical beam vector and
estimate the target's optimal azimuth angle. Then, the IRS can fix the horizontal beam based on the estimated azimuth angle, and estimate the elevation angle by efficiently tuning its vertical beam following the similar method as for the azimuth angle estimation.}
\begin{align}\label{Eq:yt2}
{\bf y}[t]&=\alpha_{\rm r} \alpha_{\rm CI} {\bf b}(\theta_{\rm IT, h}) {\bf q}^T(\theta_{\rm IT, h}) {\boldsymbol{\varphi}}_{\rm h}[t] \eta_{\rm r}+\alpha_{\rm d} {\bf b}(\theta_{\rm IT, h})+{\bf z}[t],
\end{align}
where  $ {\bf q}^T(\theta_{\rm IT, h})\triangleq{\bf u}^T(\tilde{\phi}_{\rm IT, h}, N_{\rm h})\in \mathbb{C}^{1\times N_h}$ and $\eta_{\rm r}\triangleq{\bf u}^T(\tilde{\phi}_{\rm IT, v}, N_{\rm v}){\boldsymbol{\varphi}}_{\rm v}[t]$.
For ease of notation, in the sequel, we simply re-denote $\theta_{\rm IT, h}$ by $\theta$,  $N_{\rm h}$ by $N$, $M_{\rm h}$ by $M$, $\boldsymbol{\varphi}_{\rm h}[t]$ by $\boldsymbol{\varphi}[t]$, and thus
\begin{align}\label{Eq:yt3}
{\bf y}[t]&=\gamma_{\rm r}  {\bf b}(\theta) {\bf q}^T(\theta){\boldsymbol{\varphi}}[t]+\alpha_{\rm d} {\bf b}(\theta)+{\bf z}[t]\nn\\
&={\bf b}(\theta) {f}({\boldsymbol{\varphi}[t]})+{\bf z}[t],
\end{align}
with
\begin{align}\label{fvar}
{f}({\boldsymbol{\varphi}[t]})=\gamma_{\rm r} {\bf q}^T(\theta){\boldsymbol{\varphi}}[t]+\alpha_{\rm d},
\end{align}
where $\gamma_{\rm r} =\alpha_{\rm r}  \alpha_{\rm CI}\eta_{\rm r}$.
%\begin{align}
%&= {\bf b}(\theta) {f}({\boldsymbol{\varphi}}[t])+{\bf z}[t],
%{f}({\boldsymbol{\varphi}}[t])&=\alpha_{\rm r}  {\bf a}^T(\theta)\diag({\boldsymbol{\varphi}}[t]) {\bf h}_{{\rm CI}}+\alpha_{\rm d},\\
%&= \alpha_{\rm r}  {\bf q}(\theta)^T {\boldsymbol{\varphi}}[t]+\alpha_{\rm d},
%\end{align}
%with ${\bf q}(\theta)^T\triangleq {\bf a}^T(\theta)\odot {\bf h}_{{\rm CI}}^T$.
Then, the DOA of the target (i.e., $\theta$) is estimated  from $\{{\bf y}[t]\}_{t=1}^T$ by using the MUSIC algorithm, as elaborated in the next section. Note that different from the traditional MIMO radar/phased-array radar for target sensing \cite{MIMOradar}, the IRS sensors in the proposed self-sensing IRS receive signals over
the two links from the same echo angle, i.e., $\theta$, and hence
achieving enhanced DOA estimation performance.
\section{IRS Passive Reflection Design and DOA Estimation}
In this section, we design the IRS passive reflection for estimating the azimuth angle from the target based on the MUSIC algorithm.

\subsection{DOA Estimation}
To estimate the DOA of the target by using the celebrated  MUSIC algorithm, we first stack the received echo signals at the IRS sensors from the target, i.e.,  $\{{\bf y}[t]\}_{t=1}^T$, and express them in the following matrix form
\begin{align}\label{ref_mat}
\mathbf{Y}&=[{\bf y}[1],\cdots,{\bf y}[T]]={\bf b}(\theta) [{f}({\boldsymbol{\varphi}}[1]),\cdots, {f}({\boldsymbol{\varphi}}[T])]+{\bf Z}\nn\\
&\triangleq {\bf b}(\theta) {\bf f}^H + {\bf Z},
\end{align}
where $\mathbf{Z}\triangleq[\mathbf{z}[1],\cdots, \mathbf{z}[T]]\in \mathbb{C}^{M\times T}$ and ${\bf f}^H\triangleq[{f}({\boldsymbol{\varphi}}[1]),\cdots, {f}({\boldsymbol{\varphi}}[T])] \in \mathbb{C}^{1\times T}$. Thus, the covariance matrix of $\mathbf{Y}$ is given by
\begin{align}
&\mathbf{R}_{\mathbf{Y}}=\frac{1}{T}\mathbf{Y}\mathbf{Y}^H=
\frac{1}{T}{\bf b}(\theta) {\bf f}^H\,{\bf f}{\bf b}^H(\theta) + \sigma^2{\bf I}_M,\label{Eq:RY}
\end{align}
Following the procedures of the MUSIC algorithm, the eigenvalue decomposition of $\mathbf{R}_{\mathbf{Y}}$ in \eqref{Eq:RY} is first obtained as
\begin{align}\label{rx}
\mathbf{R}_{\mathbf{Y}}=[\mathbf{U}_s, \mathbf{U}_z]\begin{bmatrix}
\boldsymbol{\Sigma}_s  &  & \\
 &  & \boldsymbol{\Sigma}_z
\end{bmatrix}
\begin{bmatrix}
 \mathbf{U}_s^H\\
 \mathbf{U}_z^H\\
\end{bmatrix}
\end{align}
where $\mathbf{U}_s\in \mathbb{C}^{M\times 1}$ and $\mathbf{U}_{z}\in \mathbb{C}^{M\times (M-1)}$ are the
eigenvectors that span the signal and noise subspaces,
respectively. As  $\mathbf{b}(\theta )$ is orthogonal to $\mathbf{U}_z$, the MUSIC spectrum parameter is given by
\begin{align}\label{music}
P_{\mathrm{MUSIC}}(\theta )=\frac{1}{\mathbf{b}(\theta )^H\mathbf{U}_z\mathbf{U}_z^H\mathbf{b}(\theta )}.
\end{align}
By finding the maximum value of \eqref{music} over $\theta$, the DOA of the target is estimated accordingly.

\subsection{IRS Reflection Design}
Although it is difficult to characterize the DOA estimation MSE of the MUSIC algorithm, it has been shown in \cite{probing} that the MSE tends to decrease with the increase of the average received signal power at the sensors. This is intuitively expected, since with a larger  signal power illuminating the target, the received SNR of the echo signals reflected by the target is also higher at the IRS sensors, thus leading to a lower DOA estimation MSE \cite{snr}.

According to \eqref{Eq:yt},
%it can be seen that the phase shift $\omega_{n,t}, \forall n, t$ of IRS in ${\boldsymbol{\varphi}}[t]$ can be optimized to
%improve the detection accuracy.
the average received signal power at IRS sensors is given by
\begin{align}
&P(\{{\boldsymbol{\varphi}[t]}\})
= \mathbb{E}\left[\frac{1}{T}\sum_{t=1}^{T}   {f}^{\dagger}({\boldsymbol{\varphi}[t]}) {\bf b}^H(\theta) {\bf b}(\theta) {f}( {\boldsymbol{\varphi}[t]}) \right]\nn\\
&=\nonumber \mathbb{E}\left[\frac{M}{T}\sum_{t=1}^{T}   {f}^{\dagger}({\boldsymbol{\varphi}[t]}) {f}( {\boldsymbol{\varphi}[t]}) \right] \\
&=\mathbb{E}\left[\frac{M}{T}\sum_{t=1}^{T}(\gamma_{\rm r}^{\dagger}{\boldsymbol{\varphi}}^H[t]{\bf q}^{\dagger}(\theta)+\alpha_{\rm d}^{\dagger}) \left(\gamma_{\rm r}{\bf q}^T(\theta )
{\boldsymbol{\varphi}}[t]+\alpha_{\rm d}\right)\right]\nonumber\\
&=\mathbb{E}[\gamma_{\rm r}\gamma_{\rm r}^{\dagger}]\text{tr}(\mathbf{R_{\boldsymbol{\varphi}}B})+{M \mathbb{E}[\alpha_{\rm d}\alpha_{\rm d}^{\dagger}]}+\frac{M}{T}\mathbb{E}[\beta_{\rm r}\beta_{\rm d}^{\dagger}]G_{\rm r}G_{\rm d}\alpha_{\rm CI}\eta_{\rm r}\nn\\
&\times{\bf q}^T(\theta )\sum_{t=1}^{T}{\boldsymbol{\varphi}}[t]+\frac{M}{T}\mathbb{E}[\beta_{\rm r}^{\dagger}\beta_{\rm d}]G_{\rm r}G_{\rm d}\alpha_{\rm CI}^{\dagger}\eta_{\rm r}^{\dagger}\sum_{t=1}^{T}{\boldsymbol{\varphi}}^H[t]{\bf q}^{\dagger}(\theta ), \label{power}
\end{align}
where the expectation is taken w.r.t. the random small-scale fading coefficients $\beta_{\rm r}$ and $\beta_{\rm d}$,
%\begin{align}\label{power_B}
$\mathbf{R}_{\boldsymbol{\varphi}}=\frac{1}{T}\sum_{t=1}^{T}{\boldsymbol{\varphi}}[t]{\boldsymbol{\varphi}}^H[t]$, and $\mathbf{B}=M{\bf q}^{\dagger}(\theta ){\bf q}^T(\theta )$. Since $\mathbb{E}[\beta_{\rm r}]=0$ and $\mathbb{E}[\beta_{\rm d}]=0$, we have
% the last two items of \eqref{power} can be ignored.
\begin{align}
&P(\{{\boldsymbol{\varphi}[t]}\})=\mathbb{E}[\gamma_{\rm r}\gamma_{\rm r}^{\dagger}]\text{tr}(\mathbf{R_{\boldsymbol{\varphi}}B})+{M \mathbb{E}[\alpha_{\rm d}\alpha_{\rm d}^{\dagger}]}. \label{power2}
\end{align}

We aim to optimize the IRS passive reflection vectors, $\{\boldsymbol{\varphi}[t]\}$, for maximizing $P(\{{\boldsymbol{\varphi}[t]}\})$, which is equivalent to maximizing $\tr(\mathbf{R_{\boldsymbol{\varphi}}B})$ for any given $\theta$. As the IRS has no prior knowledge of the target location, we aim to maximize $\text{tr}(\mathbf{R_{\boldsymbol{\varphi}}B})$ in the worst-case scenario w.r.t. $\mathbf{B}$, which is formulated as the following optimization problem
% Specifically, we have
\begin{align}
%\label{bigamp}
(\mathrm{P1}):~ &\!\!\!\!\mathop{\max}\limits_{\{\boldsymbol{\varphi}[t]\}}\mathop{\min}\limits_{\mathbf{B}}\text{tr}(\mathbf{R_{\boldsymbol{\varphi}}B})\nonumber\\
&\!\!\mathrm{s.t.}~[\mathbf{R}_{\boldsymbol{\varphi}}]_{n,n}=1, ~~n\in\mathcal{N},\\
&\!\!\!\!~~~~~~\mathbf{R}_{\boldsymbol{\varphi}}\succeq0, ~\mathbf{B}\succeq 0,\\
&\!\!\!\!~~~~~~\lambda_n(\mathbf{B})\geq \epsilon, ~~~n\in\mathcal{N},\label{Eq:B}\\
&\!\!\!\!~~~~~~ |[\boldsymbol{\varphi}[t]]_n|=1, ~~~n\in\mathcal{N}, t\in\mathcal{T},
\end{align}
where $\mathcal{N}\triangleq\{1, \cdots, N\}$, $\epsilon \geq 0$ and the constraint in \eqref{Eq:B} is imposed to eliminate the trivial solution of $\mathbf{B}=0$. Similar to \cite{snr}, it can be shown that the optimal solution to problem (P1) satisfies   $\mathbf{R}^{*}_{\boldsymbol{\varphi}}=\mathbf{I}_N$. Let $\boldsymbol{\Theta}\triangleq[{\boldsymbol{\varphi}}[1],\cdots,{\boldsymbol{\varphi}}[T]]\in \mathbb{C}^{N \times T}$ and thus we have $\mathbf{R}_{\boldsymbol{\varphi}}=\frac{1}{T}\boldsymbol{\Theta}\boldsymbol{\Theta}^H=\mathbf{I}_N$. This indicates that the optimal passive reflection matrix $\boldsymbol{\Theta}^{*}$ should be an orthogonal matrix with each entry satisfying the unit-modulus constraint. For example, one such matrix for $\boldsymbol{\Theta}^{*}$ is the matrix that concatenates the first $N$ columns of a $T\times T$ DFT matrix with $T\geq N$, with each entry given by
%$[\boldsymbol{\Theta}^{\dagger}]$ , denoted by $\boldsymbol{\Theta}=[{\boldsymbol{\varphi}}[1],\cdots,{\boldsymbol{\varphi}}[T]]\in \mathbb{C}^{N \times T}$, should be a is given by
% an optimal design of phase shift matrix that attains the maximum power is the $N$ leading columns of a $T\times T$ DFT matrix, which is given by
\begin{align}\label{the}
[\boldsymbol{\Theta}]_{n,t}=e^{-j\frac{2\pi(t-1)(n-1)}{T}}, ~n\in\mathcal{N}, t\in\mathcal{T}.
\end{align}
Note that in this case, the optimal passive reflection matrix generates an omnidirectional beampattern in the angular domain for scanning the target in all possible directions. This is expected since the omnidirectional IRS beampattern is optimal for the IRS-reflected echo link to locate the target direction, while the resultant
random phase difference between the two signals (reflected by IRS and non-reflected by IRS) arriving at the target does not affect the average estimation performance.

\section{Performance analysis}
In this section, we first show
%that in the IRS sensing system,  it is beneficial to employ the IRS controller (instead of the mobile device) for sensing probing signals.
   the performance advantage of the proposed IRS sensing system as compared to a baseline IRS sensing system assisted by a mobile device/user. Then, we characterize the CRB of the DOA estimation MSE in the considered IRS sensing system.
\subsection{Performance Gain}
%We first
\subsubsection{IRS-reflected versus direct echo links}
First,
%we separately analyze the average received signal powers of the individual reflecting and direct links, and
we show that the IRS-reflected echo link dominates the direct echo link in the average received signal power at the IRS sensors, when the number of IRS reflecting elements is sufficiently large. To this end, we characterize the average powers of the signals over these two links as follows.
\begin{lemma}\label{Lem:ReDiPow}{With the IRS reflection matrix given in \eqref{the}, the average received signal powers at the IRS sensors over the IRS-reflected and direct echo links, denoted by $P_{\rm r}$ and $P_{\rm d}$, respectively, are given by
\begin{align}\label{prpd}
P_{\rm r}=\frac{{{N M\eta_{\rm r}\lambda ^4}\kappa }}{{1024{\pi ^5}{d_{\rm IT}^4d_{{\rm CI}}^2}}}, ~~~~~  P_{\rm d}=\frac{M{{\lambda ^2}\kappa }}{{64{\pi ^3}{d_{\rm CT}^2d_{\rm IT}^2}}}.
\end{align}}
\end{lemma}
\begin{proof}
For the IRS-reflected echo link ${\bf g}_{\rm r}[t]$, the average received signal power is given by
% we apply the MUSIC algorithm, for which its the covariance matrix of reflect link is given by
\begin{align}\label{ref2}
P_{\rm r}&= \frac{1}{T}\sum_{t=1}^{T}\mathbb{E}\left[ (\gamma_{\rm r}^{\dagger}{\boldsymbol{\varphi}}^H[t]{\bf q}^{\dagger}(\theta){\bf b}^H(\theta))(\gamma_{\rm r}  {\bf b}(\theta) {\bf q}^T(\theta){\boldsymbol{\varphi}}[t])\right]\nn\\
&=\mathbb{E}[\gamma_{\rm r}\gamma_{\rm r}^{\dagger}]\text{tr}(\mathbf{R_{\boldsymbol{\varphi}}B})=NM \mathbb{E}[\gamma_{\rm r}\gamma_{\rm r}^{\dagger}]= \frac{{{N M\eta_{\rm r}\lambda ^4}\kappa }}{{1024 {\pi ^5}{d_{\rm IT}^4}d_{{\rm CI}}^2}}.
\end{align}
%\begin{align}\label{normr}
%P_{\rm r}=\frac{1}{T}\sum_{t=1}^{T}\mathbb{E}\left[\parallel\gamma_{\rm r} {\bf a}^T(\theta)\diag({\boldsymbol{\varphi}}[t]) {\bf h}_{{\rm CI}}\parallel_2^2\right]=\gamma_{\rm r}\gamma_{\rm r}^{\dagger}N=\frac{{{N\eta_{\rm r}\lambda ^4}\kappa }}{{64*16{\pi ^5}{d_{\rm IT}^4}d_{{\rm CI}}^2}},
%\end{align}
Next, for the direct echo link ${\bf g}_{\rm d}[t]$, the average received signal power is given by
\begin{align}\label{dir2}
P_{\rm d}&= \frac{1}{T}\sum_{t=1}^{T}\mathbb{E}\left[ (\alpha_{\rm d}^{\dagger} {\bf b}^H(\theta))(\alpha_{\rm d} {\bf b}(\theta))\right]=\frac{{M{\lambda ^2}\kappa }}{{64{\pi ^3}{d_{\rm CT}^2d_{\rm IT}^2}}}.
%=\mathbb{E}[\gamma_{\rm r}\gamma_{\rm r}^{\dagger}]\text{tr}(\mathbf{R_{\boldsymbol{\varphi}}B})=NM \mathbb{E}[\gamma_{\rm r}\gamma_{\rm r}^{\dagger}]= \frac{{{N M\eta_{\rm r}\lambda ^4}\kappa }}{{64\times16 {\pi ^5}{d_{\rm IT}^4}d_{{\rm CI}}^2}}.
\end{align}
%  The covariance matrix of direct link is given by
%\begin{align}\label{ref2}
%\mathbf{R}_{\rm d}=P_{\rm d}\mathbf{b}(\theta)\mathbf{b}(\theta)^H+\sigma^2\mathbf{I}.
%\end{align}
%with
%\begin{align}\label{normd}
%P_{\rm d}=\alpha_{\rm d}\alpha_{\rm d}^{\dagger}=\frac{{{\lambda ^2}\kappa }}{{64{\pi ^3}{d_{\rm CT}^2d_{\rm IT}^2}}},
%\end{align}
%\eqref{normr} and \eqref{normd} leads to Lemma 1.
\end{proof}

Lemma~\ref{Lem:ReDiPow} shows that the average power of the IRS-reflected echo link is linearly increasing with the number of IRS reflecting elements, i.e., $N$.
Moreover, based on Lemma~\ref{Lem:ReDiPow}, we obtain the following result.
\begin{lemma}\label{Lem:Nth}{
The average power of the IRS-reflected echo link exceeds that of the direct echo link when
$$N \ge N_{\rm th}\triangleq \frac{d_{\rm IT}^2d_{\rm CI}^2 16\pi^2}{\eta_r\lambda^2d_{\rm CT}^2}.$$}
\end{lemma}

Lemma~\ref{Lem:Nth} shows that when the number of IRS reflecting elements is sufficiently large, the IRS-reflected echo link leads to a more dominant received signal power in the IRS-reflected echo link than the direct echo link. Moreover, to exploit this gain, it is desirable to place the IRS controller closer to IRS reflecting elements for reducing their path-loss in the IRS controller$\to$IRS elements link.

% sho it is observed that the number of reflecting elements $N$ is proportional to the square of the controller-IRS distance $d_{\rm CI}$. This indicates that if we want to exploit the IRS for local sensing, we would better put the IRS controller near the IRS for radar signal transmission, otherwise the benefit of IRS is negligible especially when the target is far from the IRS.

\subsubsection{IRS controller versus mobile user for sending probing signals}
Next, we show that
%it is beneficial to employ the IRS controller than the mobile device for sensing probing signals.
 the proposed IRS sensing architecture by exploiting the IRS controller to send probing signals achieves better estimation performance than a benchmark system that needs a nearby mobile user to send probing signals.
Note that in this user-aided system, the estimation performance critically depends on the location of the assisting mobile user.

For ease of comparison, we consider a typical case where the mobile device locates at the $x$-axis with its position given by $(d_{\rm UI},0,0)$. As such, given the DOA of the target $\theta$ and the IRS elements-target distance $d_{\rm IT}$, the distance between the device and target is obtained as
\begin{align}\label{dct}
d_{\mathrm{UT}}=\sqrt{d_{{\rm UI}}^2+d_{\mathrm{IT}}^2-2d_{{\rm UI}}d_{\mathrm{IT}}\cos(\theta)}.
\end{align}
For convenience, we set $\theta=0$ and consider the case where the device locates between the IRS and target, i.e., $0<d_{\rm UI}<d_{\rm IT}$. Note that in \eqref{dct}, the case with a sufficiently small $d_{\rm UI}$ reduces to the proposed IRS self-sensing system.
%where ensure that the transmitter is located between the IRS and the target, i.e., $d_{\rm UI}, d_{\rm CT} \leq d_{\rm IT}$.
Based on Lemma~\ref{Lem:ReDiPow} and \eqref{dct}, we first derive the average power of the combined channel including both the IRS-reflected and direct echo links w.r.t. $d_{\rm UI}$ for the user-aided IRS sensing system as follows.
\begin{lemma}{
For the user-aided IRS sensing system with $0<d_{\rm UI}<d_{\rm IT}$, the average power of the combined channel at the IRS sensors is given by
\begin{align}\label{pc}
P_{\rm c}(d_{\rm UI})&= P_{\rm r}(d_{\rm UI})+ P_{\rm d}(d_{\rm UI})\nn\\
&=MN\eta_{\rm r}\frac{{{\lambda ^4}\kappa }}{{1024{\pi ^5}{d_{\mathrm{IT}}^4}}d_{\rm UI}^2}\nn\\
&+\frac{M{{\lambda ^2}\kappa }}{{64{\pi ^3}{d_{\mathrm{IT}}^2(d_{\rm UI}^2+d_{\mathrm{IT}}^2-2d_{\rm UI}d_{\mathrm{IT}}\cos(\phi))}}}.
\end{align}}
\end{lemma}
%\begin{proof}
%The average received signal power at IRS sensors over the combined channel is given by
%\begin{align}\label{P_c}
%&P_{\rm c}(d_{\rm UI})=\frac{1}{T}\parallel \underbrace{[\gamma_{\rm r}\mathbf{a}(\theta )^T\mathrm{diag}({\boldsymbol{\varphi}}[1])
%\mathbf{g}x(1),\cdots,\gamma_{\rm r}\mathbf{a}(\theta )^T\mathrm{diag}({\boldsymbol{\varphi}}[T])
%\mathbf{g}x(T)]}_{\text{Reflect link}}+\underbrace{[\gamma_{\rm d} x (1),\cdots,\gamma_{\rm d} x (T)]}_{\text{Direct link}}\parallel_2^2\nonumber\\
%&=\gamma_{\rm r}\gamma_{\rm r}^{\dagger}N\gamma_{h}^2+\gamma_{\rm d}^2+\frac{1}{T}\gamma_{\rm r}\gamma_{\rm d}^{\dagger}{\mathbf{a}}(\theta )^T\mathrm{diag}(\mathbf{g})\sum_{t=1}^{T}{\boldsymbol{\varphi}}[t]+\frac{1}{T}\gamma_{\rm r}^{\dagger}\gamma_{\rm d}\sum_{t=1}^{T}{\boldsymbol{\varphi}}[t]^H\mathrm{diag}(\mathbf{g})^{\dagger}\mathbf{a}(\phi )^{\dagger}.
%\end{align}
%Substituting the expressions $\gamma_{\rm d}$, $\gamma_{h}$ and $\gamma_{\rm r}$ into \eqref{P_c} and after some calculations, we obtain \eqref{pc}.
%\end{proof}
Then, the effects of $d_{\rm UI}$ on $P_{\rm r}, P_{\rm d}$, and $P_{\rm c}$ are characterized as follows.

\begin{lemma}\label{Lem:Prdc}{
For the user-aided IRS sensing system, as $d_{\rm UI}$ increases, we have
\begin{itemize}
\item $P_{\rm r}(d_{\rm UI})$ monotonically decreases;
\item $P_{\rm d}(d_{\rm UI})$ monotonically increases;
\item $P_{\rm c}(d_{\rm UI})$ first monotonically decreases and then increases.
\end{itemize}}
\end{lemma}
\begin{proof}
The first two results on $P_{\rm r}(d_{\rm UI})$ and $P_{\rm d}(d_{\rm UI})$ can be easily obtained, since the device$\to$IRS link distance increases and the device$\to$target link distance decreases.
%  derivative of $P_{\rm r}(d_{\rm UI})$ is
%$$P_{\rm r}(d_{\rm UI})'=\frac{{{-2N\eta_{\rm r}\lambda ^4}\kappa }}{{64*16{\pi ^5}{d_{\rm IT}^4}}}\frac{1}{d_{{\rm CI}}^3}\leq 0,$$
%thus $P_{\rm r}$ is monotonically decreases.
%
%The derivative of $P_{\rm d}$ is
%$$P_{\rm d}(d_{\rm UI})'=\frac{{{2\lambda^2}\kappa }}{{64{\pi ^3}d_{\rm IT}^2}}(d_{\rm IT}-d_{\rm UI})\geq 0,$$
%thus $P_{\rm d}$ is monotonically increases.
%The second derivative of $P_{\rm c}(d_{\rm UI})$ is given by
%\begin{align}\label{fr_de}
%P_{\rm c}(d_{\rm UI})''=\frac{3N\eta_r\kappa\lambda^4}{512 \pi^5d_{\mathrm{IT}}^4d_{\rm UI}^4}+\frac{\kappa\lambda^2(3d_{\rm UI}^2-6d_{\rm UI}d_{\mathrm{IT}}\cos(\phi)+4d_{\mathrm{IT}}^2\cos^2(\phi)-d_{\mathrm{IT}}^2)}{32\pi^3d_{\mathrm{IT}}^2(d_{\mathrm{IT}}^2-2\cos(\phi)d_{\mathrm{IT}}d_{\rm UI}+d_{\rm UI}^2)^3}.
%\end{align}
%Substituting $\phi=0$ into formula \eqref{fr_de}, we have $P_{\rm c}(d_{\rm UI})''>0$. So $P_{\rm c}$ is a convex function.
For the combined link, it can be shown that the first-order derivative of $P_{\rm c}(d_{\rm UI})$ is
%By solving the following equation
\begin{align}\label{pcd}
P_{\rm c}'(d_{\rm UI})=\frac{-2NM\eta_r\lambda^4\kappa(d_{\rm IT}-d_{\rm UI})^3+32M\lambda^2\kappa\pi^2d_{\rm IT}^2d_{\rm UI}^3}{1024\pi^5d_{\rm IT}^4d_{\rm UI}^3(d_{\rm IT}-d_{\rm UI})^3}.
\end{align}
Then, it follows that there exists a ${\bar d}_{\rm UI}$, such that  $P_{\rm c}'(d_{\rm UI})<0$ when  $0<d_{\rm UI}<{\bar d}_{\rm UI}$ and $P_{\rm c}'(d_{\rm UI})\geq 0$ when  ${\bar d}_{\rm UI}\le d_{\rm UI}<d_{\rm UT}$, thus leading to the desired result.
%we can get $d_{\rm UI}$ corresponding to the minimum signal power, which is denoted as $d_{\rm AI,c}$. Herein, $P_{\rm c}(d_{\rm UI})'$ denotes the first derivative of the $P_{\rm c}(d_{\rm UI})$ with respect to $d_{\rm UI}$.
%Thus $P_{\rm c}(d_{\rm UI})$ first monotonically increases in $[0, d_{\rm AI,c}]$ and then decreases in $[d_{\rm AI,c}, +\infty]$.
\end{proof}

Lemma~\ref{Lem:Prdc} shows that when the mobile device locates closer to the IRS as compared to the target, the IRS-reflected echo link yields a much larger power than the direct echo link in the average due to the prominent IRS passive beamforming gain. In contrast, when the mobile user gets closer to the target, the channel gain of the direct echo link is greatly enhanced due to the smaller path-loss from the user to the target, while the IRS-reflected echo link is severely attenuated.
\begin{example}\label{Exa:P_MSE}{In Figs.~\ref{power} and ~\ref{PL}, we plot the DOA estimation MSE by the MUSIC algorithm and average received signal power at the IRS sensors, respectively.
%When direct link power and reflect link power cumulated at sensors are equal, the closed expression of $d_{\mathrm{CI}}$ is given by
%\begin{align}\label{er}
%d_{\mathrm{CI,c}}=\frac{d_{\mathrm{IT}}}{\frac{4\pi d_{\mathrm{IT}}}{\sqrt{N\eta_r}\lambda}+1}.
%\end{align}
%In Fig.~\ref{power},
We set $d_{\mathrm{IT}}=30$ m, $N=64$.
% and $d_{\mathrm{CI}}=0.47$.
 It is observed that for the case with $d_{\rm UI}=0.47$ m, the direct echo link and IRS-reflected echo link achieve the same signal power at the sensors and thus the same MSE performance, which is consistent with Lemmas~\ref{Lem:ReDiPow} and \ref{Lem:Nth}.
In addition, it is observed that as $d_{\rm UI}$ increases, $P_{\rm c}$ first decreases and then increases.
% This is because the received power of the reflecting link
In contrast, the DOA estimation MSE  first increases and then decreases with $d_{\rm UI}$. By letting $P_{\rm c}'(d_{\rm UI})=0$ (see \eqref{pcd}), we obtain that $P_{\rm c}(d_{\rm UI})$ achieves its minimum value when $d_{\rm UI}=1.7852$ m, which also achieves the maximum MSE.}
% is matched with the maximum MSE point. Note that $d_{\rm AI,c}$ is close to $d_{\rm CI,c}$.
\end{example}

\begin{figure}[t]
\centering
\includegraphics[width=75mm]{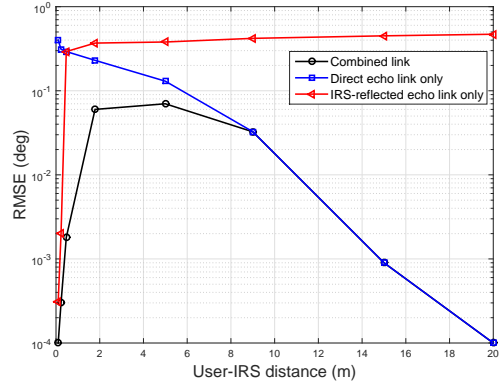}
\caption{Effect of $d_{\rm UI}$ on the DOA estimation MSE.}
\label{power}
\end{figure}

\begin{figure}[t]
\centering
\includegraphics[width=75mm]{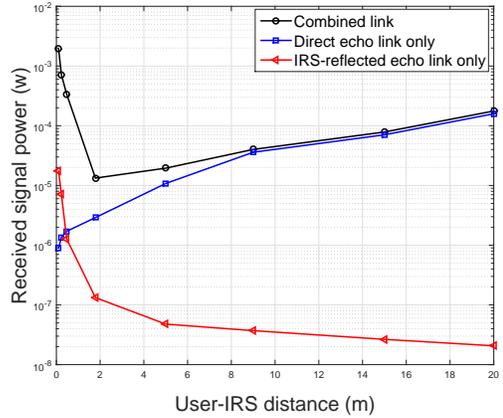}
\caption{Effect of $d_{\rm UI}$ on the received signal power.}
\label{PL}
\end{figure}

\begin{remark}{
%{\color{blue}When the signal transmitter is placed near the IRS or near the target, the system has a higher detection accuracy, and then it can sense a longer distance. However, the device cannot guarantee that it will always be near the target, but the controller can always guarantee that it is near the IRS.}
Example~\ref{Exa:P_MSE} shows that it is beneficial to employ the helping mobile device near either the IRS or the target for target localization, such that the DOA estimation MSE is minimized. Note that the former case exploits the IRS passive beamforming gain, while the latter case mainly relies on the short-distance user-target link for DOA estimation. However, in practice, since the availability of such helping user is random and so is its location, the estimation performance of the user-aided IRS sensing system is not guaranteed. In contrast, using the IRS controller as the probing signal transmitter (which is equivalent to the case of short user-IRS distance in Example 1) avoids this issue and thus offers a more predictable estimation performance than the user-aided IRS sensing system.}
%Moreover, the target localization in the anchor-based IRS sensing requires the additional information of the anchor location, which introduces
\end{remark}
\subsection{Cramer-Rao Bound}
Since the DOA estimation MSE of the MUSIC algorithm is difficult to obtain, we analyze in this subsection the CRB of the proposed IRS self-sensing system that characterizes a lower bound of the DOA estimation MSE. Moreover, we gain useful insights into the effects of the number of IRS reflecting elements or sensors on the estimation performance.

%People aim to find the optimal estimator that minimizes the Mean Square Error (MSE)of DOA estimation. However,
%such an optimal estimator is often difficult to construct and
%the minimum MSE (MMSE) is normally hard to characterize.
%To evaluate the performance of the DOA estimation, To gain more insights on the performance limits and the effect of $M$ and $N$ on estimation performance,
%we consider the well know Cramer-Rao Bound (CRB) that characterizes the lower bound of MSE estimation.

First, we obtain the following result of the CRB for the target DOA estimation accuracy in the proposed IRS self-sensing system.

\begin{theorem}\label{Theorem1}{For the proposed IRS self-sensing system, the CRB of the target DOA estimation MSE is given by \eqref{theorem1} at the top of next page,
\begin{figure*}[t]
\begin{align}\label{theorem1}
%\!\!\!\!\!\!\!\mathrm{CRB}(\theta )\!=\!\frac{1}{2\frac{T|\xi\alpha_{\rm{CI}}\eta_{\rm r}|^2}{\sigma^2}\!\left(\!w_1\frac{M^3-M}{12}\cos^2(\theta)N\!+\!\pi^2w_2\frac{N^3-N}{12}\cos^2(\theta)M\!+\!w_1\frac{M^3-M}{12}\cos^2(\theta)\frac{1}{\alpha_{\rm CI}\alpha_{\rm CI}^*\eta_{\rm r}\eta_{\rm r}^*}\!\right)\!},\\
\mathrm{CRB}(\theta )\!=\!\frac{1/2}{\frac{T|\xi\alpha_{\rm{CI}}\eta_{\rm r}|^2}{\sigma^2}\!\left(\!w_2\frac{N^3-N}{6}\varsigma^2(\theta)M+\!w_1\frac{M^3-M}{6}\cos^2(\theta)N\!+ \!w_1\frac{(M^3-M)}{6}\frac{\cos^2(\theta)}{\alpha_{\rm CI}\alpha_{\rm CI}^{\dagger}\eta_{\rm r}\eta_{\rm r}^{\dagger}}\!+\!w_1\frac{(M^3-M)}{6}\cos^2(\theta)p(\theta)\r)\!}.\!\!
\end{align}
\hrulefill
\end{figure*}
where
\begin{align*}
w_1\!&=\!\frac{\pi^2d_{\rm s}^2}{\lambda^2},~~w_2\!=\!\frac{\pi^2d_{\rm I}^2}{\lambda^2},~~\xi\!=\!\max(\alpha_{\rm r},\alpha_{\rm d}),\nn \\
{\varsigma}(\theta)\!&=\!\cos(\theta)\sin(\theta_{\rm IT,v})\!+\sin(\theta_{\rm CI, h})\sin(\theta_{\rm CI, v}),\nn \\
p(\theta)\!&=\!\Re\left\{\frac{2}{T\alpha_{\rm CI}^{\dagger}\eta_{\rm r}^{\dagger}}\sum_{t=1}^{T}\!\left({{\boldsymbol{\varphi}}[t]_1e^{j(N-1)\pi\frac{{\rm d_I}}{\lambda}\bar{\varsigma}(\theta)}}+\cdots \right.\right.\nn\\
&\left.\left.+{{\boldsymbol{\varphi}}[t]_Ne^{-j(N-1)\pi\frac{{\rm d_I}}{\lambda}\bar{\varsigma}(\theta)}}\!\right)\!\right\},\nn\\
\bar{\varsigma}(\theta)&=\sin(\theta)\sin(\theta_{\rm IT,v})\!+\sin(\theta_{\rm CI, h})\sin(\theta_{\rm CI, v}),
\end{align*}
with ${\boldsymbol{\varphi}}[t]_n$ being the $n$th element of vector ${\boldsymbol{\varphi}}[t]$.}
\end{theorem}

\begin{proof}
Please refer to Appendix A.
\end{proof}

\begin{remark}
Theorem~\ref{Theorem1} shows  that the CRB of the DOA estimation MSE decreases with both the number of IRS sensors, $M$, and that of IRS reflecting elements, $N$. Specifically, in the denominator of \eqref{theorem1}, the first two terms show the effects of the IRS-reflected echo link on the CRB performance, while the third term (i.e., $\frac{T|\xi|^2}{\sigma^2}(\!w_1\frac{M^3-M}{6}\cos^2(\theta))$) shows the effect of the direct echo link. The last term shows the effect of the IRS-reflected echo link and the direct echo link.
Such results reveal that when $|\alpha_{\rm{CI}}\eta_{\rm r}|^2$ is small, a large number of IRS reflecting elements (i.e., a large $N$) is needed to compensate the severe path-loss due to the two-hop signal reflections in the IRS-reflected echo link. On the other hand, increasing the number of IRS sensors can improve the performance as it is beneficial for boosting the estimation SNR of the received signals over both the direct and IRS-reflected echo links.
%Moreover, increasing the number of sensors is easier to improve the accuracy of estimation, compared to increasing the number of IRS elements. This is due to the fact that,
%in \eqref{theorem1}, in addition to the term introduced by direct link, $w_1\frac{M^3-M}{12}\cos^2(\theta)N+\pi^2w_2\frac{N^3-N}{12}\cos^2(\theta)M$, the direct link introduced another term $w_1\frac{M^3-M}{12}\cos^2(\theta)\frac{1}{\alpha_{\rm CI}\alpha_{\rm CI}^*\eta_{\rm r}\eta_{\rm r}^*}$, which is only related to the number of sensors $M$. In this context, a small number of sensors can
%be set for sensing, as will be verified by simulations in
%Section V, where the DOA detection error can approach
%zero with the existence of limited sensors.
\end{remark}

\section{Numerical Results}
\begin{table*}[!t]
\caption{{Simulation parameters}}
\label{Table1}
\centering
\begin{tabular}{|c|c|}
\hline
{\bf{Parameter}} & {\bf{Value}} \\
\hline
Number of IRS reflecting elements & $N=64$ \\
\hline
Number of IRS sensors & $M=8$ \\
\hline
Distance from IRS to target & $d_{\rm IT}=30$ m \\
\hline
Angle of target w.r.t. IRS  & $\theta=60^{\circ}$ \\
\hline
Target RCS & $\kappa=7$ dBsm \\
\hline
Noise power & $-109$ dBm \\
\hline
Wavelength & $\lambda=0.2$ m \\
\hline
Number of BS transmit antennas & $64$ \\
\hline
Number of BS receive antennas & $8$ \\
\hline
Distance from IRS to BS & $d_{\rm BI}=100$ m \\
\hline
Angle of BS w.r.t. IRS & $\theta_{\rm I}=80^{\circ}$ \\
\hline
Angle of IRS w.r.t. BS & $\theta_{\rm B}=80^{\circ}$ \\
\hline
Distance from IRS controller to reflecting elements & $d_{\rm CI}=0.5$ m \\
\hline
Successful angle estimation threshold & $\delta=0.01$ \\
\hline
Number of snapshots & $T=64$ \\
\hline
Angle of device w.r.t. IRS & $60^{\circ}$ \\
\hline
Distance from helping user to IRS & Uniformly distributed in $[0.5,100]$ m or $[0.5,135]$ m\\
\hline
\end{tabular}
\end{table*}
\begin{figure*}[!t]
  \centering
  \subfigure[BTB scheme]{
    \label{fig:subfig:a} %% label for first subfigure
    \includegraphics[width=2.3in]{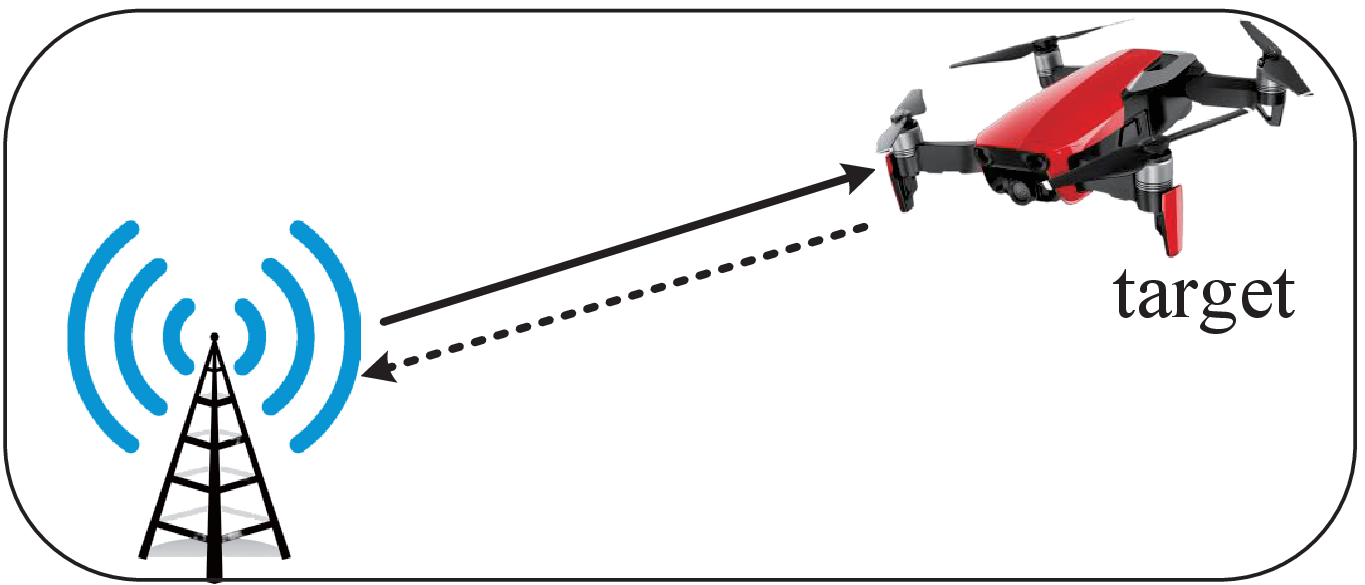}}
  \hspace{1in}
  \subfigure[BITS scheme]{
    \label{fig:subfig:b} %% label for second subfigure
    \includegraphics[width=2.3in]{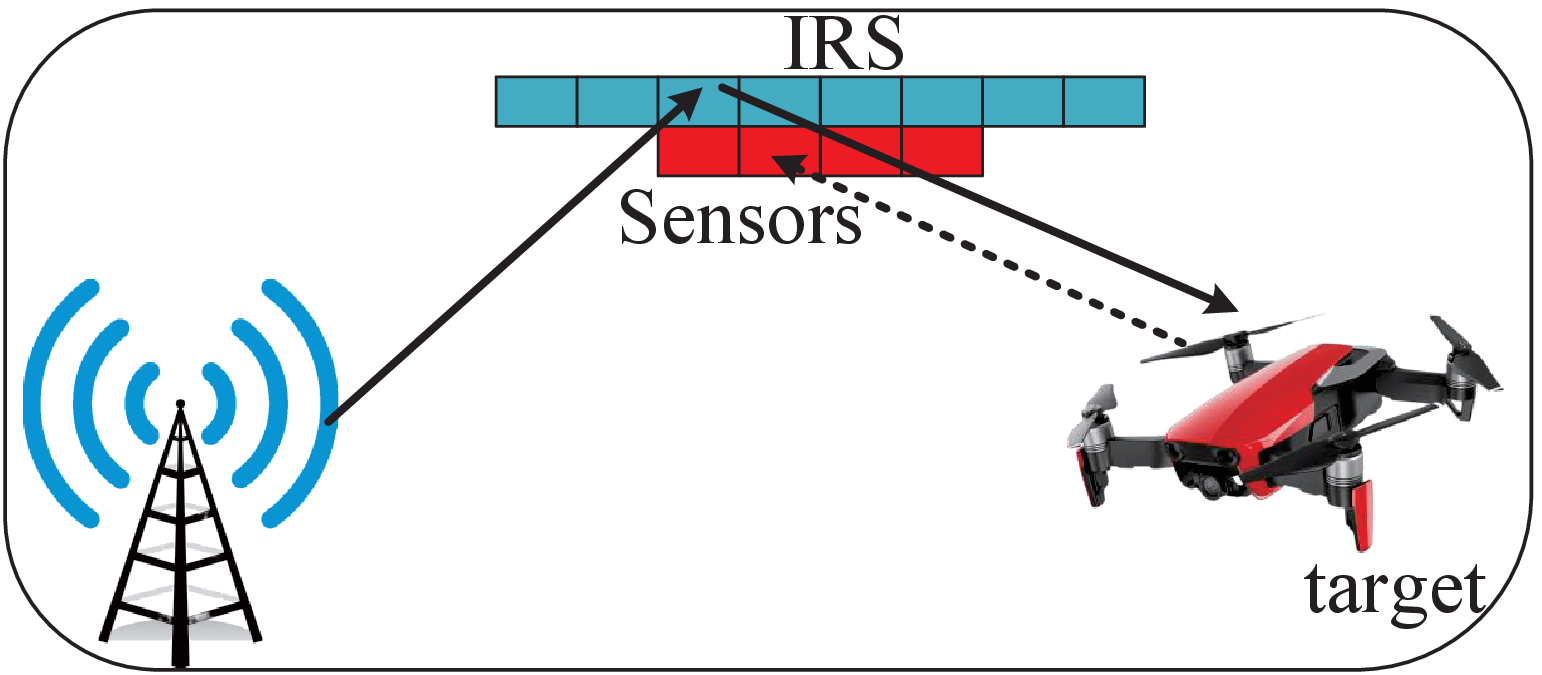}}
  \label{fig:subfig} %% label for entire figure
    \hspace{1in}\\
  \subfigure[BTS scheme]{
    \label{fig:subfig:b} %% label for second subfigure
    \includegraphics[width=2.3in]{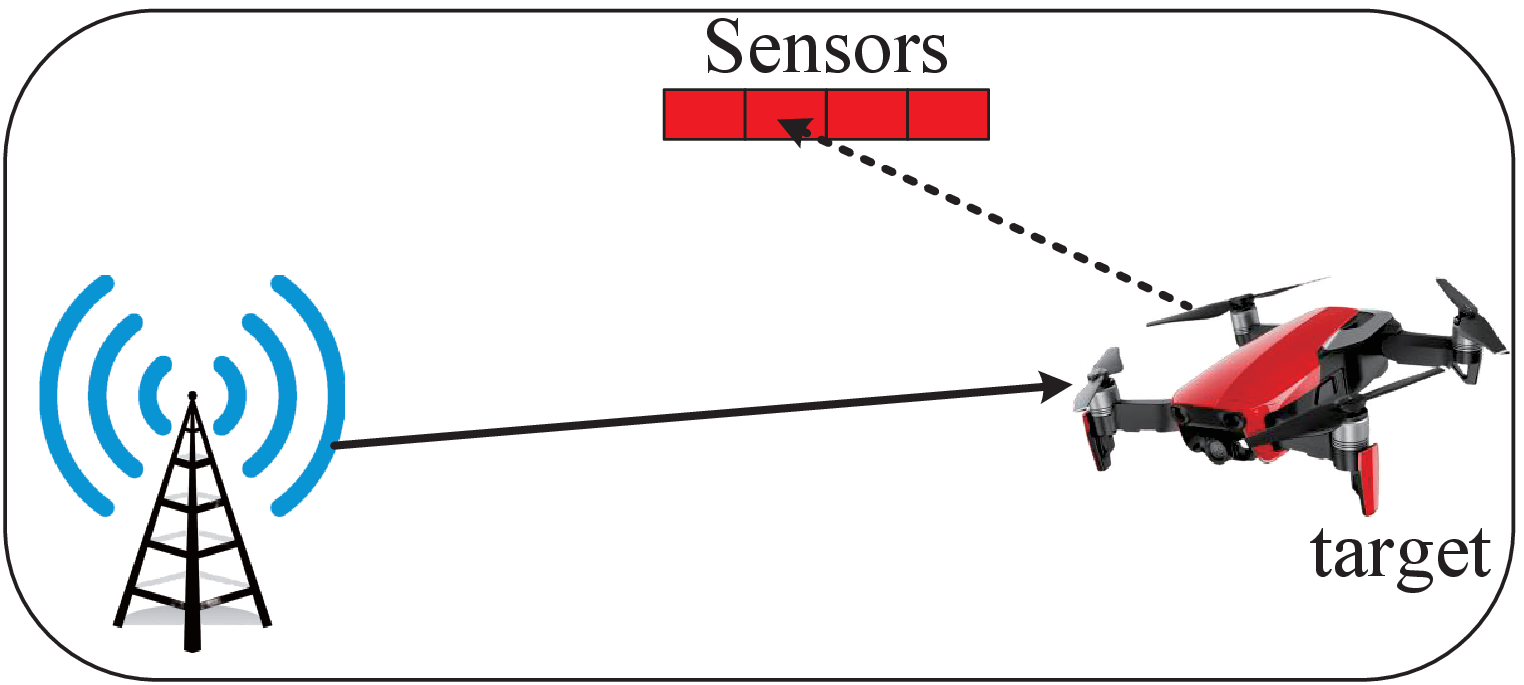}}
  \label{fig:subfig} %% label for entire figure
  \hspace{0.001in}
  \subfigure[BITIB scheme]{
    \label{fig:subfig:b} %% label for second subfigure
    \includegraphics[width=2.3in]{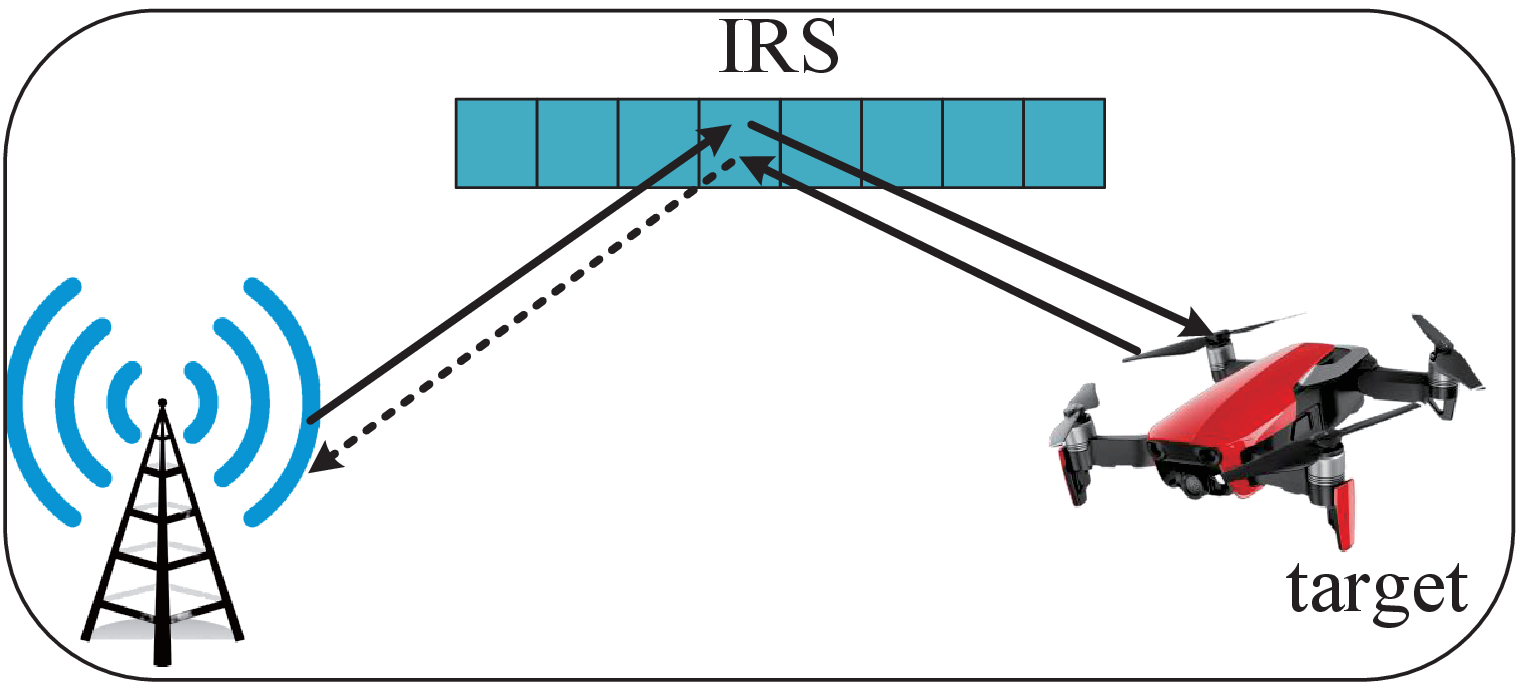}}
  \label{fig:subfig} %% label for entire figure
     \hspace{0.001in}
       \subfigure[Mobile-user aided scheme (MUS)]{
    \label{fig:subfig:b} %% label for second subfigure
    \includegraphics[width=1.7in]{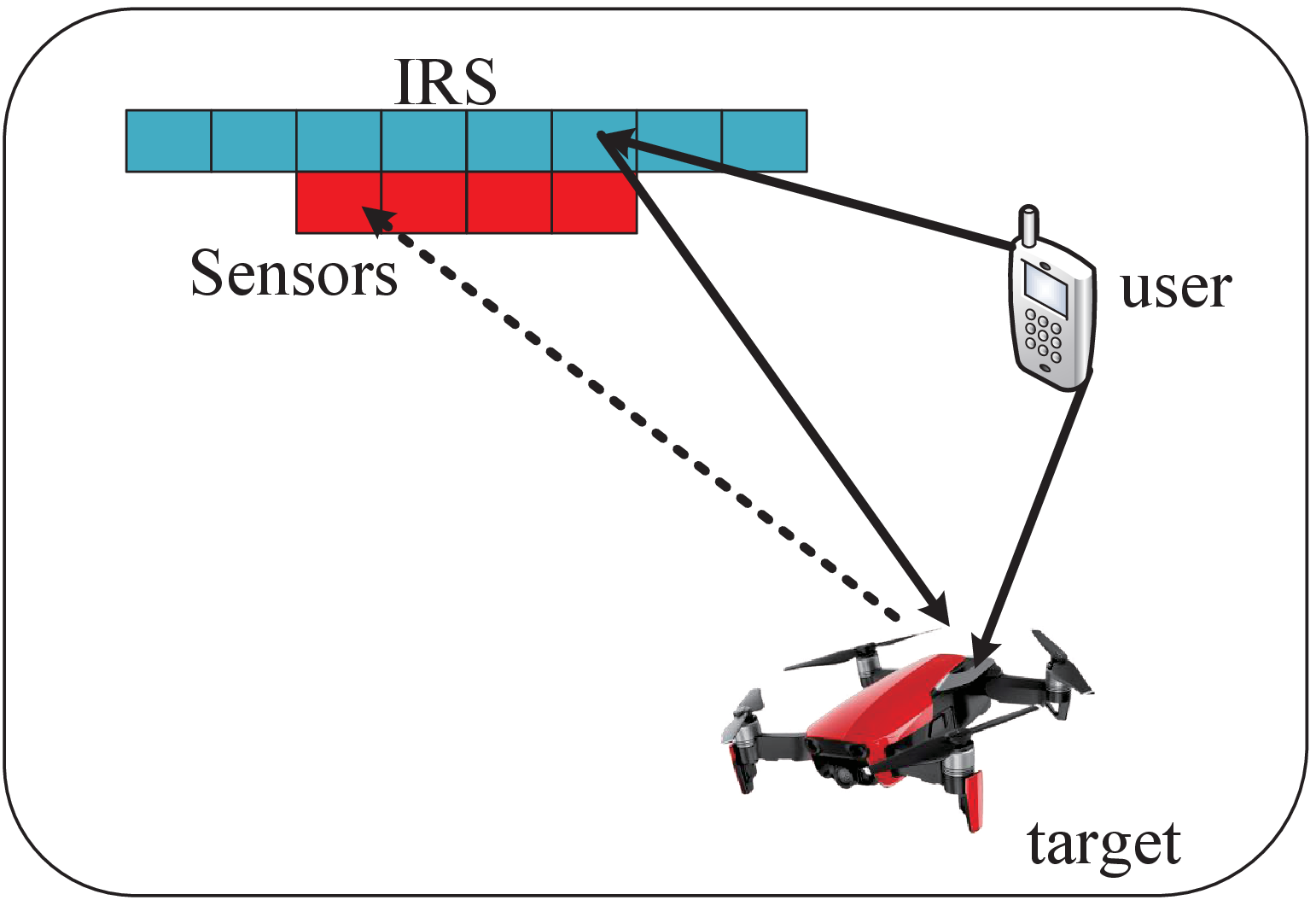}}
  \label{fig:subfig} %% label for entire figure
    \caption{Different benchmark sensing systems with/without IRS.}
    \label{comparision_scene}
\end{figure*}

In this section, we present numerical results to evaluate the performance of the proposed IRS-self-sensing system. For simplicity, we consider a ULA IRS as shown in Fig.~\ref{Simulation1}, where the IRS is equipped with $N=64$ reflecting elements and $M=8$ sensors.
The distance between the IRS and target is $d_{\rm IT}=30$ m, the azimuth angle of the target w.r.t. the IRS is set as $\theta=60^{\circ}$. Moreover, for the benchmark schemes detailed in the sequel,
%sensing performance comparison,
we consider there is a BS equipped with $64$ transmit antennas and $8$ receive antennas, where the distance between the IRS and BS is set as $d_{\rm BI}=100$ m.
%the azimuth angle of the BS is $\theta_{\rm I}=80^{\circ}$, and the azimuth angle of the IRS to the BS is set to be $\theta_{\rm B}=10^{\circ}$.
The other simulation parameters are set according to Table~\ref{Table1}, unless stated otherwise. It is worth noting that the IRS controller in our considered system is practically located in the far-field of `each IRS element' due to its small size as composed to the distance with the controller.
The DOA estimation performance  is evaluated in terms of  root mean square error (RMSE) and probability of successful estimation. Specifically, RMSE is defined as $\textrm{RMSE} \triangleq \sqrt{\mathbb{E}[(\theta-\hat{\theta})^2]}$, where $\hat{\theta}$ is the estimation of $\theta$. A target is said to be localized successfully in a given trial if $|\hat{\theta}-\theta|\leq \delta$ with a tunable constant $\delta$ \cite{succ1,succ2}. Herein, $\delta$ is set to be $0.01$.

We consider the following benchmark schemes that include the traditional radar system and the schemes that apply IRS in different ways to form the echo links to estimate the DOA of the target.
\begin{itemize}
\item[1)] BTB scheme: As illustrated in Fig.~\ref{comparision_scene}(a), the BS transmits and receives probing signals for target sensing; thus the echo link is BS$\to$target$\to$BS.
\item[2)] BITS scheme: As illustrated in Fig.~\ref{comparision_scene}(b), the BS transmits probing signals, which are reflected by the IRS and received by its sensors; thus the echo link is BS$\to$IRS elements$\to$target$\to$ IRS sensors.
\item[3)] BTS scheme: As illustrated in Fig.~\ref{comparision_scene}(c), the BS transmits probing signals, which are reflected by the target and received by its sensors; thus the echo link is BS$\to$target$\to$ IRS sensors.
\item[4)]  BITIB scheme:  As illustrated in Fig.~\ref{comparision_scene}(d), the BS transmits probing signals, which are consecutively reflected over the BS$\to$IRS elements$\to$target$\to$IRS elements$\to$BS echo link.
\item[5)] Mobile-user aided  scheme (MUS): This scheme, as illustrated in Fig.~\ref{comparision_scene}(e) and previously considered in Example 1, is similar to our  proposed scheme, while except that the probing signals are sent by a mobile user (instead of the IRS controller) whose location is random in practice.
% , the IRS controller is replaced by the
\end{itemize}
Note that for the  BITIB scheme, the DOA of target is estimated by the on-grid beam training method due to the lack of IRS sensors; while the MUSIC algorithm is applied to estimate the target DOA in the other benchmark schemes.
%introduce five sensing methods to compare with our proposed IRS aided sensing method. As shown in , the first comparison method uses the BS to transmit and receive signal . The flow of the signal is from the BS to the target and then to the BS (BTB).
%
% The second comparison method uses the BS to transmit signal, which is reflected by the IRS and received by the sensors. The flow of the signal is from the BS to the IRS, and then to the target and finally to the sensors (BITS).
%
%  The third comparison method also uses the BS to transmit signal, which is reflected by the target and received by the sensors. The flow of the signal is from the BS to the target and finally to the sensors (BTS).
%
% The fourth comparison method uses the BS to transmit signal, which is reflected by the IRS twice and received by the BS. The flow of the signal is from the BS to the IRS to the target, and then to the IRS and finally to the BS (BITIB).
%
% The final comparison method uses the device to transmit signal. The signal is reflected by the IRS to the target and is also directly transmitted to the target. Both two signals transmitted to the target are received by the sensors. The flow of the signal is from the device to the IRS to the target then to the sensors, and is also from the device to the target then to the sensors (MUS).

In Fig. \ref{Fig1}, we show the RMSE of the proposed
IRS self-sensing scheme versus BS transmit power. It is observed that the proposed scheme exhibits
superior performance than the other benchmark schemes in the entire transmit power range.
%Such performance improvement mainly
%benefits from the extended DOFs of the IRS.
This is because, in the proposed IRS self-sensing scheme, the IRS is placed near the target with its controller and sensors for  transmitting and receiving signals, respectively. As such, the signal traveling distance from the transmitter to the target and then to the receiver is much shorter than that of the BTB, BITS, BTS, and  BITIB schemes. Note that in these four benchmark schemes, the signal is sent by the far-away BS, thus suffering severe path-loss. Moreover, the proposed IRS scheme can receive signals over both the direct and IRS-reflected echo links at the same time. As stated in Lemma 3, the proposed scheme with the combined links performs better than those based on the individual link only. Thus, the proposed IRS-self-sensing scheme that  exploits both the IRS passive beamforming gain and the direct echo link gain for target DOA estimation is practically appealing.
Besides, the proposed scheme substantially improves the RMSE performance as compared to
the MUS scheme, since the latter scheme
cannot guarantee that there always exists a helping device near the IRS/target.
\begin{figure}[t]
\centering
\includegraphics[width=85mm]{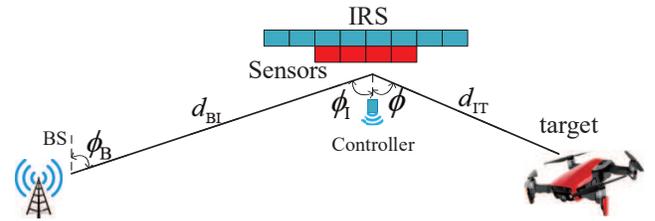}
\caption{{Simulation setup.}}
\label{Simulation1}
\end{figure}

In Fig.~\ref{Fig2}, we evaluate the probability of successful estimation of the proposed scheme versus the transmit power against the other benchmark schemes. It is observed that the proposed scheme attains a much higher successful estimation probability than the other schemes. Moreover, the successful estimation probability of all schemes increases with the transmit power due to the increased received SNR at the BS/sensors. It is observed that when the distance from helping user to IRS is short, e.g., [0.5, 100] m,
the proposed scheme performs worse than the MUS scheme in terms of probability of successful estimation in the low transmit power region, while it attains better performance when the transmit power increases. The reason is that when the transmit power is small, our proposed IRS self-sensing system has a low probability of successful estimation due to the low receive SNR at the sensors; while the MUS scheme has the potential to attain a higher probability of successful estimation, since the user has some likelihood to be located in the vicinity of the target which leads to higher DOA estimation accuracy.
On the other hand, in the high-transmit power region, the proposed IRS self-sensing system achieves a high probability of successful estimation at high SNR, while the MUS scheme suffers some performance loss, since the user may be far away from the target and/or IRS and thus become the performance bottleneck.

\begin{figure}[t]
\centering
\includegraphics[width=80mm]{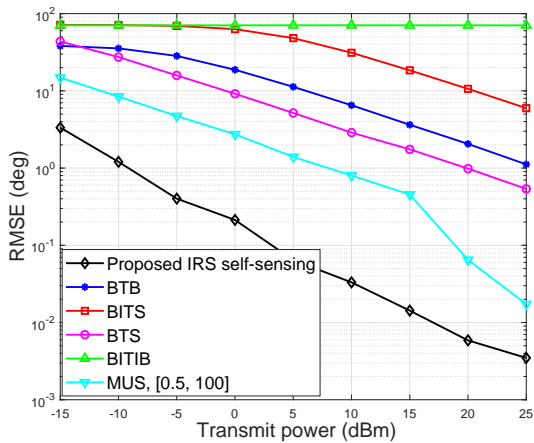}
\caption{RMSE versus transmit power.}
\label{Fig1}
\end{figure}

\begin{figure}[t]
\centering
\includegraphics[width=80mm]{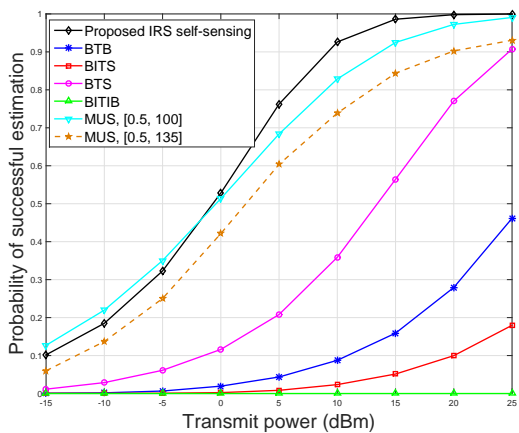}
\caption{Probability of successful estimation versus transmit power.}
\label{Fig2}
\end{figure}

\begin{figure}[t]
\centering
\includegraphics[width=80mm]{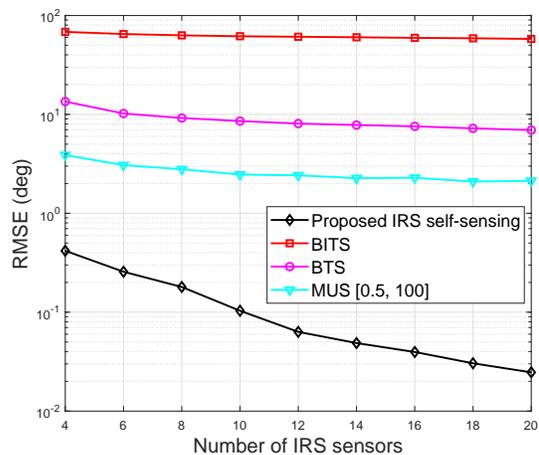}
\caption{RMSE versus number of IRS sensors.}
\label{Fig9}
\end{figure}

\begin{figure}[t]
\centering
\includegraphics[width=80mm]{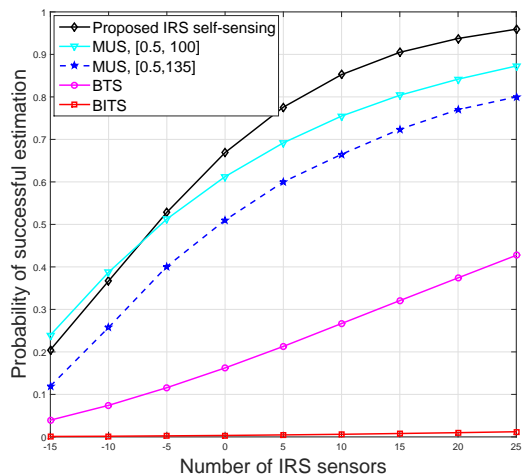}
\caption{Probability of successful estimation versus number of IRS sensors.}
\label{Fig10}
\end{figure}

Next, we investigate the impact of the number of sensors on the performance of the proposed IRS-self-sensing scheme in Fig. \ref{Fig9} and Fig. \ref{Fig10}. It is observed that the estimation performance is critically dependent on  the number of IRS sensors, i.e., $M$. Specifically,
an increasing $M$ leads to a decreasing NMSE and an increasing probability of successful estimation for all schemes. One can observe that
the proposed IRS-self-sensing scheme requires a small number of sensors for achieving a high  DOA estimation accuracy as compared to the BITS, BTS, and MUS schemes that all use IRS sensors to receive signals.
Note that there also exists a threshold on the number of sensors, above which the proposed IRS self-sensing system achieves a higher probability of successful
estimation. This is fundamentally due to the increasing average SNR at the sensors with the increasing number of IRS sensors; thus similar arguments as for the effect of transmit power (see Fig. \ref{Fig2}) apply. Moreover, the performance gap of the proposed scheme against other benchmark schemes is
enlarged as the number of sensors increases.

\begin{figure}[t]
\centering
\includegraphics[width=80mm]{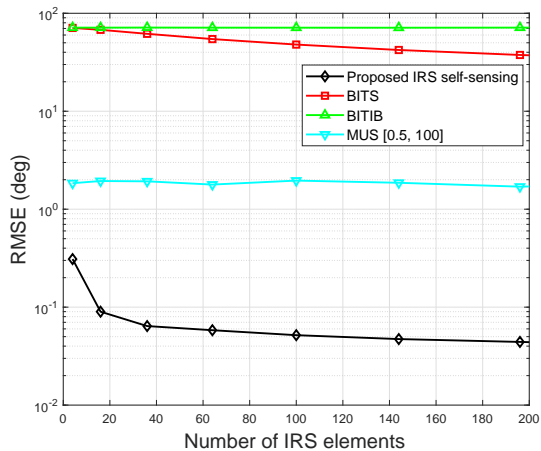}
\caption{RMSE versus number of IRS reflecting elements.}
\label{Fig7}
\end{figure}

\begin{figure}[t]
\centering
\includegraphics[width=80mm]{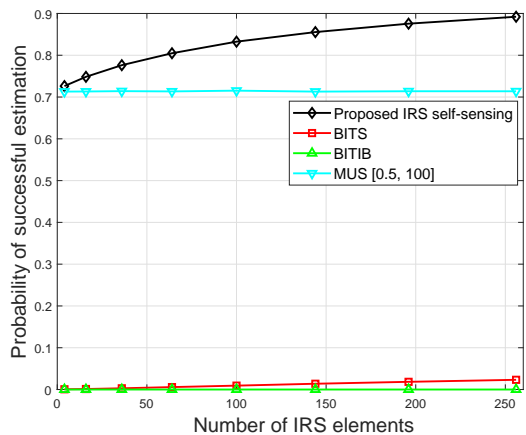}
\caption{Probability of successful estimation versus number of IRS reflecting elements}.
\label{Fig8}
\end{figure}

In Fig. \ref{Fig7} and Fig. \ref{Fig8}, we show the impact of the number of IRS reflecting elements on the sensing performance. It is observed that the proposed scheme obviously outperforms the schemes that sends the signal from BS or user. For the proposed scheme, the NMSE monotonically decreases with the IRS reflecting elements, i.e., $N$ and the probability of successful estimation monotonically increases with $N$. The reason is that as the number
of IRS reflecting elements increases, the IRS passive beamforming
gain continues to increase.
However, the sensing performance of benchmark schemes does not improve too much when $N$ increases, since for those schemes, the IRS passive beamforming gain cannot greatly compensate the path-loss, especially in the case of far-away BS$\to$ IRS link.
It is observed that as $N$ increases, the RMSE (or the probability of successful estimation) of the proposed scheme
first rapidly decreases (increases) and then more slowly when $N$ is sufficiently large.

\begin{figure}[t]
\centering
\includegraphics[width=81mm]{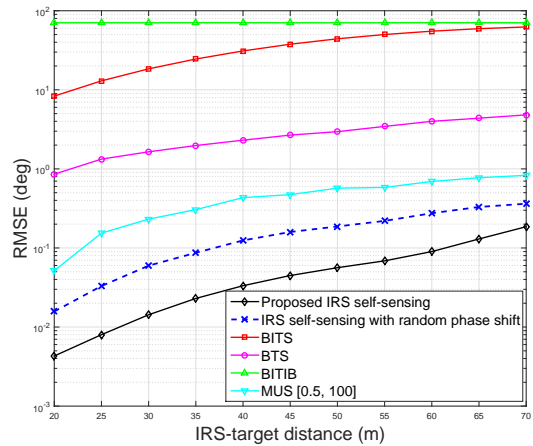}
\caption{RMSE versus distance between IRS/sensors and target.}
\label{Fig5}
\end{figure}

\begin{figure}[t]
\setlength{\abovecaptionskip}{-0.cm}
\setlength{\belowcaptionskip}{0.cm}
  \centering
\includegraphics [width=81mm]{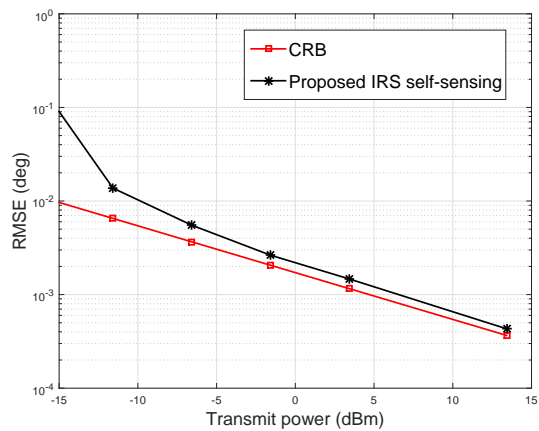}
\caption{CRB and IRS self-sensing performance for DOA estimation.}
\label{CRBfig}
\end{figure}

In Fig.~\ref{Fig5}, we show the NMSE performance of different sensing
schemes versus the IRS/sensors-target distance. It is observed that the proposed scheme achieves a small NMSE even at a long sensing range, thanks to the IRS controller in short distance with the reflecting elements. In contrast, the other benchmark schemes suffer a poor NMSE performance when the target is far from the IRS/sensors. Since the signal traveling distance from the transmitter to the receiver is too long, the average received signal power at IRS sensors is extremely low and the target cannot be sensed at long distances. Moreover, compared with the random-phase scheme that randomly generates IRS reflection over time, our proposed IRS reflection design for target sensing achieves a much smaller RMSE, since it yields an omnidirectional beampattern for target sensing.

%since the IRS controller short controller-IRS distance achieves a higher sensing accuracy.
%We can also observe that when the target number increases, the MUS
%breaks down faster than the proposed scheme.

Last, we compare in Fig.~\ref{CRBfig} the RMSE of the proposed IRS-self-sensing scheme by using the MUSIC algorithm and its CRB in Theorem~\ref{Theorem1}. It is observed that although there is a notable  gap in the low transmit power regime, the achieved NMSE of the proposed scheme gets very close to the obtained CRB when the transmit power is sufficiently large. This demonstrates the effectiveness of the proposed IRS-self-sensing scheme.

\section{Conclusion}
In this paper,  we proposed a new IRS-self-sensing system, where the IRS controller is employed to transmit probing signals,
and dedicated sensors are installed at the IRS for location/angle estimation based on the reflected signals by the target with and without the IRS reflection. The MUSIC algorithm was applied to estimate the DOA of the target in the IRS's vicinity with high accuracy, without the involvement of either the BS or any mobile device. Although the DOA estimation MSE by the MUSIC algorithm is intractable, we optimized the IRS passive reflection matrix for maximizing the average received signals' total power at the IRS sensors, which leads to the minimum MSE.
Besides, we analytically showed that it is beneficial to employ the IRS controller (instead of a helping mobile device at random location) for sending probing signals, and obtained the CRB for the target DOA estimation.

\begin{appendices}
\section{Proof of Theorem~1}\label{App:SimuLoS}
We first introduce an auxiliary variable $\mathbf{c}=[0,\cdots,\frac{1}{\alpha_{\rm{CI}}\eta_r},\cdots,0]^T\in \mathbb{R}^{N\times 1}$. Then, \eqref{Eq:yt3} can be rewritten as
\begin{align}\label{rew_y}
{\bf y}[t]&=\alpha_{\rm r}\alpha_{\rm{CI}}\eta_{\rm r} {\bf b}(\theta) {\bf q}^T(\theta){\boldsymbol{\varphi}}[t]+\alpha_{\rm d}  {\bf b}(\theta)+{\bf z}[t]\nonumber\\
&\overset{(a_1)}{=} \alpha_{\rm{CI}}\eta_{\rm r}{\bf b}(\theta) {\bf q}^T(\theta)\left(\alpha_{\rm r}{\boldsymbol{\varphi}}[t]+\alpha_{\rm d}\mathbf{c} \right)+{\bf z}[t],
\end{align}
where we use the property $\alpha_{\rm{CI}}\eta_{\rm r}{\bf q}^T(\theta){\bf c}=1$ in $(a_1)$. Without loss of generality, by setting
$\xi=\max(\alpha_{\rm r},\alpha_{\rm d})$, \eqref{rew_y} can be approximated as
\begin{align}\label{appro_y}
{\bf y}[t]\approx\xi\alpha_{\rm{CI}}\eta_{\rm r} {\bf b}(\theta) {\bf q}^T(\theta)\left({\boldsymbol{\varphi}}[t]+\mathbf{c} \right)+{\bf z}[t].
\end{align}
Since  ${\bf z}[t]$ is CSCG  distributed, the received signal at the sensors, ${\bf y}[t]$, can be modeled as an independent CSCG vector, i.e., ${\bf y}[t]\sim\mathcal{CN}(\xi \alpha_{\rm{CI}}\eta_{\rm r}{\bf b}(\theta) {\bf q}^T(\theta)\left({\boldsymbol{\varphi}}[t]+\mathbf{c} \right),\sigma^2\mathbf{I})$.
\newcounter{TempEqCnt}
\setcounter{TempEqCnt}{\value{equation}}
\setcounter{equation}{48}
\begin{figure*}[ht]
\begin{align}\label{CRB1}
\mathrm{CRB}(\theta )=\frac{1}{2\frac{T|\xi\alpha_{\rm{CI}}\eta_{\rm r}|^2}{\sigma^2}\left({\bf q}^H(\theta)\mathbf{R}^T_{\rm c}{\bf q}(\theta)\left\|\dot{\bf b}(\theta)\right\|_2^2+M\dot{\bf q}^H(\theta)\mathbf{R}^T_{\rm c}\dot{\bf q}(\theta)-\frac{M|{\mathbf{q}}^H(\theta)\mathbf{R}^T_{\rm c}\dot{\mathbf{q}}(\theta)|^2}{{\bf q}^H(\theta)\mathbf{R}^T_{\rm c}{\bf q}(\theta)}\right)}.
\end{align}
\hrulefill
\end{figure*}
\setcounter{equation}{\value{TempEqCnt}}

\newcounter{TempEqCnt1}
\setcounter{TempEqCnt1}{\value{equation}}
\setcounter{equation}{52}
\begin{figure*}[ht]
\begin{align}\label{g3}
&{\bf q}^H(\theta)\mathbf{R}^T_{\rm c}{\bf q}(\theta)\left\|\dot{\bf b}(\theta)\right\|_2^2 = 4\sum_{m=1,2,\cdots,\frac{M-1}{2}}\pi^2\frac{d_{\rm s}^2}{\lambda^2}m^2\cos^2(\theta)\left(N+\frac{1}{\alpha_{\rm CI}\alpha_{\rm CI}^{\dagger}\eta_{\rm r}\eta_{\rm r}^{\dagger}}\right.\nn \\
&\left.+\underbrace{\Re\left\{\frac{2}{T\alpha_{\rm CI}^{\dagger}\eta_{\rm r}^{\dagger}}\sum_{t=1}^{T}\left({{\boldsymbol{\varphi}}[t]_1e^{j(N-1)\pi\frac{{\rm d_I}}{\lambda}\bar{\varsigma}(\theta)}}+\cdots+{{\boldsymbol{\varphi}}[t]_Ne^{-j(N-1)\pi\frac{{\rm d_I}}{\lambda}\bar{\varsigma}(\theta)}}\right)\right\}}_{p(\theta)}\right)\nn\\
&=\frac{M^3-M}{6}\pi^2\frac{d_{\rm s}^2}{\lambda^2}\cos^2(\theta)\left(N+\frac{1}{\alpha_{\rm CI}\alpha_{\rm CI}^{\dagger}\eta_{\rm r}\eta_{\rm r}^{\dagger}}+p(\theta)\right),
\end{align}
\hrulefill
\end{figure*}
\setcounter{equation}{\value{TempEqCnt}}

Let $\boldsymbol{\varpi}=[\theta, \bar{\boldsymbol{\xi}}]^T$ denote the vector of unknown parameters to be estimated, which includes the DOA and the complex amplitudes with $\bar{\boldsymbol{\xi}}=[\Re\{\xi\},\Im\{\xi\}]^T$.
The log-likelihood function for estimating $\boldsymbol{\varpi}$ from  $\mathbf{Y}=({\bf y}[1],\cdots,{\bf y}[T])$ can be derived as
\begin{align}\label{logy}
\log f_{\mathbf{Y}}(\mathbf{Y};\boldsymbol{\varpi})=\bar{f}_1(\mathbf{Y})+\bar{f}_2(\boldsymbol{\varpi})
+\bar{f}_3(\mathbf{Y},\boldsymbol{\varpi}),
\end{align}
where
%$\bar{f}_1(\mathbf{Y})$ given by
%\begin{align}\label{f1y}
$\bar{f}_1(\mathbf{Y})=-MN\log(\pi\sigma^2)-\frac{1}{\sigma^2}\sum_{t=1}^{T}\left \| {\bf y}[t] \right \|_2^2$,
%\end{align}
and $\bar{f}_2(\boldsymbol{\varpi})$ is given by
\begin{align}\label{f2y}
&\bar{f}_2(\boldsymbol{\varpi})=-\frac{1}{\sigma^2}\sum_{t=1}^{T}\left \| \xi\alpha_{\rm{CI}}\eta_{\rm r} {\bf b}(\theta) {\bf q}^T(\theta)\left({\boldsymbol{\varphi}}[t]+\mathbf{c} \right) \right \|_2^2\nonumber\\
&\!=\!-\frac{1}{\sigma^2}\sum_{t=1}^{T}\xi^{\dagger}\xi\alpha_{\rm{CI}}^{\dagger}\alpha_{\rm{CI}}\eta_{\rm r}^{\dagger}\eta_{\rm r}\!\left(\!{\boldsymbol{\varphi}}^H[t]\!+\!\mathbf{c}^H \!\right)\!\mathbf{A}^H(\theta)\mathbf{A}(\theta)\left({\boldsymbol{\varphi}}[t]\!+\!\mathbf{c} \right)\nonumber\\
&=-\frac{T}{\sigma^2}\xi^{\dagger}\xi\alpha_{\rm{CI}}^{\dagger}\alpha_{\rm{CI}}\eta_{\rm r}^{\dagger}\eta_{\rm r}\text{tr}(\mathbf{A}(\theta)\mathbf{R}_c\mathbf{A}^H(\theta)),
\end{align}
where $\mathbf{A}(\theta)= {\bf b}(\theta) {\bf q}^T(\theta)$ and
\begin{align}\label{rc}
\mathbf{R}_c=\frac{1}{T}\sum_{t=1}^{T}\left({\boldsymbol{\varphi}}[t]+\mathbf{c} \right)\left({\boldsymbol{\varphi}}^H[t]+\mathbf{c}^H \right).
\end{align}
Moreover, $\bar{f}_3(\mathbf{Y},\boldsymbol{\varpi})$ is given by
\begin{align}\label{f3y}
&\bar{f}_3(\mathbf{Y},\boldsymbol{\varpi})\!=\!\frac{2}{\sigma^2}\Re\!\left\{\!\xi^{\dagger}\alpha_{\rm{CI}}^{\dagger}\eta_{\rm r}^{\dagger}\sum_{t=1}^{T}\!\left(\!{\boldsymbol{\varphi}}^H[t]+\mathbf{c}^H \!\right)\!{\bf q}^{\dagger}(\theta){\bf b}(\theta)^H{\bf y}[t]\!\right\}\nonumber\\
&\!=\!\frac{2\sqrt{T}}{\sigma^2}\Re\!\left\{\!\xi^{\dagger}\alpha_{\rm{CI}}^{\dagger}\eta_{\rm r}^{\dagger}\sum_{n=1}^{N}\mathbf{A}_{:,n}^H(\theta)\left(\frac{1}{\sqrt{T}}\sum_{t=1}^{T}{\bf y}[t]\!\left(\!{\boldsymbol{\varphi}}^{\dagger}_n[t]\!+\!\mathbf{c}^{\dagger}_n \!\right)\!
\right)\!\right\},
\end{align}
where $\mathbf{c}^{\dagger}_n$ denotes the $n$-th element of vector $\mathbf{c}^{\dagger}$.
It is observed from \eqref{f3y} that
\begin{align}\label{suf_v}
\bar{\bf v}_n=\frac{1}{\sqrt{T}}\sum_{t=1}^{T}{\bf y}[t]\left({\boldsymbol{\varphi}}^{\dagger}_n[t]+\mathbf{c}^{\dagger}_n \right)
, \forall n\in\mathcal{N},
\end{align}
is the sufficient statistic \cite{suff}, which contains all information relevant to $\boldsymbol{\varpi}$ derived from data $\mathbf{Y}$. In order to reduce the dimensionality and computational complexity, instead of using log-likelihood function in \eqref{logy}, we derive the fisher information matrix (FIM) \cite{fish} in the following sufficient statistic matrix:
\begin{align}\label{suf_m}
{{\bf V}}=[\bar{\bf v}_1,\cdots,\bar{\bf v}_N]=\frac{1}{\sqrt{T}}\sum_{t=1}^{T}{\bf y}[t]\left({\boldsymbol{\varphi}}^H[t]+\mathbf{c}^H \right).
\end{align}
To make ${\{\bar{\bf v}}_n\}_{n=1}^{N}$ statistically independent, we
decompose the matrix from \eqref{rc} by using singular value decomposition (SVD) as $\mathbf{R}_c=\mathbf{U}\boldsymbol{\Gamma }\mathbf{U}$, where $\mathbf{U}$ and $\boldsymbol{\Gamma}$ are the matrices of eigenvectors
and eigenvalues of $\mathbf{R}_c$, respectively. As such, the vector $\boldsymbol{\chi}={\boldsymbol{\varphi}}[t]+\mathbf{c} $ can be written as a linear transformation of independent signals, which is given by
\begin{align}\label{su}
\widetilde{\boldsymbol{\chi}}=\boldsymbol{\Gamma}^{-\frac{1}{2}}\mathbf{U}^H\left({\boldsymbol{\varphi}}[t]+\mathbf{c} \right).
\end{align}
Substituting \eqref{su} into \eqref{suf_m} and vectorizing it, we have
\begin{align}\label{y_sim}
{\bf v}&=\text{vec}\left(\frac{1}{\sqrt{T}}\sum_{t=1}^{T}{\bf y}[t]\left({\boldsymbol{\varphi}}^H[t]+\mathbf{c}^H \right)\mathbf{U}\boldsymbol{\Gamma}^{-\frac{1}{2}}\right)\nonumber\\
&=\xi\alpha_{\rm{CI}}\eta_{\rm r}\sqrt{T}\text{vec}\left({\bf b}(\theta) {\bf q}^T(\theta)\mathbf{U}\boldsymbol{\Gamma}^{\frac{1}{2}}\right)\nn\\
&+\text{vec}\l(\frac{1}{\sqrt{T}}\sum_{t=1}^{T}\mathbf{z}[t]
\left({\boldsymbol{\varphi}}^H[t]+\mathbf{c}^H \right)\mathbf{U}\boldsymbol{\Gamma}^{-\frac{1}{2}}\r)\nonumber\\
&=\xi\alpha_{\rm{CI}}\eta_{\rm r}\underbrace{\sqrt{T}\text{vec}\left({\bf b}(\theta) {\bf q}^T(\theta)\mathbf{U}\boldsymbol{\Gamma}^{\frac{1}{2}}\right)}_{{\rm \mathbf{w}}(\theta)}+\bar{\mathbf{z}},
\end{align}
where $\bar{\mathbf{z}}=\text{vec}(\frac{1}{\sqrt{T}}\sum_{t=1}^{T}\mathbf{z}[t]
\left({\boldsymbol{\varphi}}^H[t]+\mathbf{c}^H \right)\mathbf{U}\boldsymbol{\Gamma}^{-\frac{1}{2}})$ obeys the CSCG  distribution with zero mean and variance $\sigma^2\mathbf{I}$. $\mathbf{w}(\theta)$ denotes the array response at the direction $\theta$, and hence the last equation of \eqref{y_sim} represents the equivalent model to \eqref{appro_y}.

%The FIM \cite{fish} is an intrinsic property of the model $\log f_{\mathbf{Y}}(\mathbf{Y};\boldsymbol{\varpi})$.
Next, the FIM for estimating the parameters $\boldsymbol{\varpi}=[\theta, \bar{\boldsymbol{\xi}}]^T$ from the sufficient statistic \eqref{y_sim} is given by
\begin{align}\label{fish}
\mathbf{F}({\bf v})=\begin{bmatrix}
{f}_{\theta\theta} & \mathbf{f}_{\theta\bar{\boldsymbol{\xi}}}   \\
 \mathbf{f}^T_{\theta\bar{\boldsymbol{\xi}}}&   \mathbf{f}_{\bar{\boldsymbol{\xi}}\bar{\boldsymbol{\xi}}}
\end{bmatrix},
\end{align}
and the CRB for the $\theta$ estimation can be written as
\begin{align}\label{CRB}
{\rm CRB}(\theta)=\left[{f}_{\theta\theta}-\mathbf{f}_{\theta\bar{\boldsymbol{\xi}}}
\mathbf{f}_{\bar{\boldsymbol{\xi}}\bar{\boldsymbol{\xi}}}^{-1}\mathbf{f}^T_{\theta\bar{\boldsymbol{\xi}}}\right]^{-1},
\end{align}
where ${f}_{\theta\theta}\in \mathbb{R}^{1}$, $\mathbf{f}_{\theta\bar{\boldsymbol{\xi}}}^T\in \mathbb{R}^{2}$ and $\mathbf{f}_{\bar{\boldsymbol{\xi}}\bar{\boldsymbol{\xi}}}\in \mathbb{R}^{2\times 2}$.
Since the array origin defined in \eqref{array} is at the array centroid, we have
\begin{align}\label{symme}
&{\bf q}^H(\theta)\dot{{\bf q}}(\theta)=0,~~~{\bf b}^H(\theta)\dot{{\bf b}}(\theta)=0,\\
& \dot{\bf q}^H(\theta){{\bf q}(\theta)}=0,~~~
\dot{\bf b}^H(\theta){{\bf b}(\theta)}=0,~\forall \theta.
\end{align}
Leveraging this orthogonality property,  ${f}_{\theta\theta}$, $\mathbf{f}_{\theta\bar{\boldsymbol{\xi}}}$ and $\mathbf{f}_{\bar{\boldsymbol{\xi}}\bar{\boldsymbol{\xi}}}$ can be calculated as
\begin{align}\label{cc}
{f}_{\theta\theta}&=\frac{2}{\sigma^2}\Re\left\{\xi^{\dagger}\alpha_{\rm{CI}}^{\dagger}\eta_{\rm r}^{\dagger}\dot{{\bf w}}^H(\theta)\xi\alpha_{\rm{CI}}\eta_{\rm r}\dot{{\bf w}}(\theta)\right\}\nonumber\\
&=\frac{2T}{\sigma^2}\Re\left\{\xi\xi^{\dagger}\alpha_{\rm{CI}}\alpha_{\rm{CI}}^{\dagger}\eta_{\rm r}\eta_{\rm r}^{\dagger}\text{tr}\left(\dot{\mathbf{A}}(\theta)\mathbf{R}_{\rm c}\dot{\mathbf{A}}^H(\theta)\right)\right\}\nonumber\\
&=\frac{2T|\xi\alpha_{\rm{CI}}\eta_{\rm r}|^2}{\sigma^2}\text{tr}\left(\left[\dot{\bf b}(\theta) {\bf q}^T(\theta)+{\bf b}(\theta) \dot{\bf q}^T(\theta)\right]\mathbf{R}_{\rm c}\right.\nn\\
&\left.\times\left[\dot{\bf b}(\theta)^{\dagger} {\bf q}^H(\theta)+{\bf b}(\theta)^{\dagger} \dot{\bf q}^H(\theta)\right]\right)\!=\!\frac{2T|\xi\alpha_{\rm{CI}}\eta_{\rm r}|^2}{\sigma^2}\!\nonumber\\
&\times\left(\!{\bf q}^H(\theta)\mathbf{R}^T_{\rm c}{\bf q}(\theta)\left\|\dot{\bf b}(\theta)\right\|_2^2\!+\!M\dot{\bf q}^T(\theta)\mathbf{R}_{\rm c}\dot{\bf q}^{\dagger}(\theta)\!\right),
\end{align}
\begin{align}\label{cx}
\mathbf{f}_{\theta\bar{\boldsymbol{\xi}}}&=\frac{2}{\sigma^2}\Re\left\{\xi^{\dagger}\alpha_{\rm{CI}}^{\dagger}\eta_{\rm r}^{\dagger}\dot{{\bf w}}^H(\theta)\left((1,j)\otimes\alpha_{\rm{CI}}\eta_{\rm r}{{\bf w}}(\theta)\right)\right\}\nonumber\\
&=\frac{2T}{\sigma^2}\Re\left\{\xi^{\dagger}\alpha_{\rm{CI}}^{\dagger}\eta_{\rm r}^{\dagger}\alpha_{\rm{CI}}\eta_{\rm r}\text{tr}\left({\mathbf{A}}(\theta)\mathbf{R}_{\rm c}\dot{\mathbf{A}}^H(\theta)\right)(1,j)\right\}\nonumber\\
&=\frac{2TM\alpha_{\rm{CI}}^{\dagger}\eta_{\rm r}^{\dagger}\alpha_{\rm{CI}}\eta_{\rm r}}{\sigma^2}\Re\left\{\xi^{\dagger}{\mathbf{q}}^H(\theta)\mathbf{R}^T_{\rm c}\dot{\mathbf{q}}(\theta)(1,j)\right\},
\end{align}
\begin{align}\label{xx}
\mathbf{f}_{\bar{\boldsymbol{\xi}}\bar{\boldsymbol{\xi}}}&=\frac{2T}{\sigma^2}\Re\left\{((1,j)\otimes\alpha_{\rm{CI}}\eta_{\rm r}{{\bf w}}(\theta))^H((1,j)\otimes\alpha_{\rm{CI}}\eta_{\rm r}{{\bf w}}(\theta))\right\}\nonumber\\
&=\frac{2T}{\sigma^2}\Re\left\{\alpha_{\rm{CI}}^{\dagger}\eta_{\rm r}^{\dagger}\alpha_{\rm{CI}}\eta_{\rm r}(1,j)^H(1,j)\text{tr}\left({\mathbf{A}}(\theta)\mathbf{R}_{\rm c}{\mathbf{A}}^H(\theta)\right)\right\}\nonumber\\
&=\frac{2T}{\sigma^2}M\alpha_{\rm{CI}}^{\dagger}\eta_{\rm r}^{\dagger}\alpha_{\rm{CI}}\eta_{\rm r}{\bf q}^H(\theta)\mathbf{R}^T_{\rm c}{\bf q}(\theta)\mathbf{I}.
\end{align}
Substituting \eqref{cc}--\eqref{xx} into \eqref{CRB} yields \eqref{CRB1} at the top of previous page.

For $\mathbf{R}_c$ in \eqref{rc}, we first perform the following transformation
\setcounter{equation}{49}
\begin{align}\label{rct}
\mathbf{R}_c\!=\!\frac{1}{T}\sum_{t=1}^{T}{\boldsymbol{\varphi}}[t]{\boldsymbol{\varphi}}^H[t]\!+\!\mathbf{c} \mathbf{c}^H\!+\!\frac{1}{T}\sum_{t=1}^{T}{\boldsymbol{\varphi}}[t]\mathbf{c}^H+\frac{1}{T}\sum_{t=1}^{T}\mathbf{c}{\boldsymbol{\varphi}}^H[t].
\end{align}
Then, inserting \eqref{rct} into \eqref{CRB1} and using the property derived in Section III. B that $\mathbf{R}_{\boldsymbol{\varphi}}=\frac{1}{T}\sum_{t=1}^{T}{\boldsymbol{\varphi}}[t]{\boldsymbol{\varphi}}^H[t]=\mathbf{I}$, we have
\begin{align}\label{g1}
\frac{M|{\mathbf{q}}^H(\theta)\mathbf{R}^T_{\rm c}\dot{\mathbf{q}}(\theta)|^2}{{\bf q}^H(\theta)\mathbf{R}^T_{\rm c}{\bf q}(\theta)}=0.
\end{align}
Since $\dot{\bf q}^H(\theta)(\mathbf{c}\mathbf{c}^H)^T\dot{\bf q}(\theta)=\dot{\bf q}^H(\theta)(\frac{1}{T}\sum_{t=1}^{T}{\boldsymbol{\varphi}}[t]\mathbf{c}^H)^T\dot{\bf q}(\theta)=\dot{\bf q}^H(\theta)(\frac{1}{T}\sum_{t=1}^{T}\mathbf{c}{\boldsymbol{\varphi}}^H[t])^T\dot{\bf q}(\theta)=0$, we obtain that
\begin{align}\label{g2}
&M\dot{\bf q}^H(\theta)\mathbf{R}^T_{\rm c}\dot{\bf q}(\theta)=M\dot{\bf q}^H(\theta)\mathbf{R}_{\boldsymbol{\varphi}}^T\dot{\bf q}(\theta)\nn\\
&\!=\!4M\!\!\!\sum_{n=1,\cdots,\frac{N-1}{2}}\!\!\!n^2\pi^2\frac{d_{\rm I}^2}{\lambda^2}{\varsigma}^2(\theta)\!=\!\frac{N^3-N}{6}\pi^2\frac{d_{\rm I}^2}{\lambda^2}{\varsigma}^2(\theta)M,
\end{align}
where ${\varsigma}(\theta)=\cos(\theta)\sin(\theta_{\rm IT,v})+\sin(\theta_{\rm CI, h})\sin(\theta_{\rm CI, v})$. After some calculation, we know that ${\bf q}^H(\theta)\mathbf{R}_{\boldsymbol{\varphi}}^T{\bf q}(\theta)=N$ and ${\bf q}^H(\theta)(\mathbf{c}\mathbf{c}^H)^T{\bf q}(\theta)\!=\!\frac{1}{\alpha_{\rm{CI}}^{\dagger}\eta_{\rm r}^{\dagger}\alpha_{\rm{CI}}\eta_{\rm r}}$. As such, we finally can state that \eqref{g3} at top of previous page,
where $\bar{\varsigma}(\theta)=\sin(\theta)\sin(\theta_{\rm IT,v})+\sin(\theta_{\rm CI, h})\sin(\theta_{\rm CI, v})$.
Substituting \eqref{g1}--\eqref{g3} into \eqref{CRB1} leads to the desired result in \eqref{theorem1}. The proof of Theorem 1 is thus completed.
\end{appendices}

% we have
%\begin{align}\label{CRB2}
%\mathrm{CRB}(\theta )=\frac{1}{2\sum_{m=1,2,\cdots,\frac{M-1}{2}}\pi^2\frac{d_{\rm s}^2}{\lambda^2}m^2\cos^2(\theta)\left(N+\frac{1}{\alpha_{\rm CI}\alpha_{\rm CI}^{\dagger}\eta_{\rm r}\eta_{\rm r}^{\dagger}}\right)+2M\sum_{n=1,\cdots,\frac{N-1}{2}}n^2\pi^2\frac{d_{\rm I^2}}{\lambda^2}\sin^2(\theta_{\rm CI,v})\cos(\theta)^2}\nonumber\\
%\end{align}


\begin{thebibliography}{1}
%\bibitem{saad2019vision}
%W.~Saad, M.~Bennis, and M.~Chen, ``A vision of {6G} wireless systems:
%  Applications, trends, technologies, and open research problems,'' \emph{IEEE
%  Netw.}, vol.~34, no.~3, pp. 134--142, Mar. 2019.

\bibitem{6G1}
K. B. Letaief, W. Chen, Y. Shi, J. Zhang, and Y. A. Zhang, ``The roadmap to 6G: AI empowered wireless networks," \emph{IEEE Commun. Mag.}, vol. 57, no. 8, pp. 84-90, Aug. 2019.

\bibitem{6G2}
M. Z. Chowdhury, M. Shahjalal, S. Ahmed, and Y. M. Jang, ``6G wireless communication systems: Applications, requirements, technologies, challenges, and research directions," \emph{IEEE Open J. Commun. Society}, vol. 1, pp. 957-975, 2020.

\bibitem{cui2021integrating}
Y.~Cui, F.~Liu, X.~Jing, and J.~Mu, ``Integrating sensing and communications
  for ubiquitous IoT: Applications, trends and challenges,'' \emph{arXiv
  preprint arXiv:2104.11457}, 2021.

\bibitem{liu2020joint}
F.~Liu, C.~Masouros, A.~P. Petropulu, H.~Griffiths, and L.~Hanzo, ``Joint radar
  and communication design: Applications, state-of-the-art, and the road
  ahead,'' \emph{IEEE Trans. Commun.}, vol.~68, no.~6, pp. 3834--3862, Jun.
  2020.

\bibitem{ma2020joint}
D.~Ma, N.~Shlezinger, T.~Huang, Y.~Liu, and Y.~C. Eldar, ``Joint
  radar-communication strategies for autonomous vehicles: Combining two key
  automotive technologies,'' \emph{IEEE Signal Process. Mag.}, vol.~37, no.~4,
  pp. 85--97, Apr. 2020.

\bibitem{feng2020joint}
Z.~Feng, Z.~Fang, Z.~Wei, X.~Chen, Z.~Quan, and D.~Ji, ``Joint radar and
  communication: A survey,'' \emph{China Commun.}, vol.~17, no.~1, pp. 1--27,
  Jan. 2020.

%\bibitem{liu2021survey}
%A.~Liu, Z.~Huang, M.~Li, Y.~Wan, W.~Li, T.~X. Han, C.~Liu, R.~Du, D.~T.~K. Pin,
%  J.~Lu \emph{et~al.}, ``A survey on fundamental limits of integrated sensing
%  and communication,'' \emph{arXiv preprint arXiv:2104.09954}, 2021.

\bibitem{mono}
R. J. Burkholder, L. J. Gupta, and J. T. Johnson, ``Comparison of monostatic and bistatic radar images,'' \emph{IEEE Ante. Propa. Mag.}, vol. 45, no. 3, pp. 41-50, June 2003.

\bibitem{liangradar}
L. Liu and S. Zhang, ``A two-stage radar sensing approach based on MIMO-OFDM technology," in proc \emph{IEEE Globecom Workshops,} Dec. 2020, pp. 1-6.

\bibitem{liaskos2018new}
C.~Liaskos, S.~Nie, A.~Tsioliaridou, A.~Pitsillides, S.~Ioannidis, and
  I.~Akyildiz, ``A new wireless communication paradigm through
  software-controlled metasurfaces,'' \emph{IEEE Commun. Mag.}, vol.~56, no.~9,
  pp. 162--169, Sep. 2018.

\bibitem{wu2021intelligent}
Q.~Wu, S.~Zhang, B.~Zheng, C.~You, and R.~Zhang, ``Intelligent reflecting
  surface aided wireless communications: A tutorial,'' \emph{IEEE Trans.
  Commun.}, vol.~69, no.~5, pp. 3313--3351, May 2021.

\bibitem{basar19_survey}
E.~{Basar}, M.~{Di Renzo}, J.~{De Rosny}, M.~{Debbah}, M.~{Alouini}, and
  R.~{Zhang}, ``Wireless communications through reconfigurable intelligent
  surfaces,'' \emph{IEEE Access}, vol.~7, pp. 116\,753--116\,773, Aug. 2019.

%\bibitem{di2020smart_JSAC}
%M.~{Di Renzo}, A.~{Zappone}, M.~{Debbah}, M.~S. {Alouini}, C.~{Yuen}, J.~{de
%  Rosny}, and S.~{Tretyakov}, ``Smart radio environments empowered by
%  reconfigurable intelligent surfaces: How it works, state of research, and the
%  road ahead,'' \emph{IEEE J. Sel. Areas Commun.}, vol.~38, no.~11, pp.
%  2450--2525, Nov. 2020.

\bibitem{irs1}
H. Liu, X. Yuan, and Y. -J. A. Zhang, ``Matrix-calibration-based cascaded channel estimation for reconfigurable intelligent surface aided multiuser MIMO," \emph{IEEE J. Sel. Areas Commun.,} vol. 38, no. 11, pp. 2621-2636, Nov. 2020.

\bibitem{irs2}
S. Xia and Y. Shi, ``Intelligent reflecting surface for massive device connectivity: Joint activity detection and channel estimation," \emph{IEEE Intern. Conf. Acoustics, Speech and Signal Process. (ICASSP)}, Barcelona, Spain, 2020, pp. 5175-5179.

\bibitem{irs3}
W. Yuan, Z. Wei, S. Li, and D. W. K. Ng, ``Integrated sensing and communication-assisted orthogonal time frequency space transmission for vehicular networks," arXiv preprint arXiv:2105.03125, 2021.

\bibitem{JR:wu2018IRS}
Q.~Wu and R.~Zhang, ``Intelligent reflecting surface enhanced wireless network
  via joint active and passive beamforming,'' vol.~18, no.~11, pp. 5394--5409,
  Nov. 2019.

\bibitem{huang2018largeRIS}
C.~Huang, A.~Zappone, G.~C. Alexandropoulos, M.~Debbah, and C.~Yuen,
  ``Reconfigurable intelligent surfaces for energy efficiency in wireless
  communication,'' \emph{IEEE Trans. Wireless Commun.}, vol.~18, no.~8, pp.
  4157--4170, Aug. 2019.

\bibitem{yu2020robust}
X.~Yu, D.~Xu, Y.~Sun, D.~W.~K. Ng, and R.~Schober, ``Robust and secure wireless
  communications via intelligent reflecting surfaces,'' \emph{IEEE J. Sel.
  Areas Commun.}, vol.~38, no.~11, pp. 2637--2652, Nov. 2020.

\bibitem{you2019progressive}
C.~You, B.~Zheng, and R.~Zhang, ``Channel estimation and passive beamforming
  for intelligent reflecting surface: Discrete phase shift and progressive
  refinement,'' \emph{IEEE J. Sel. Areas Commun.}, vol.~38, no.~11, pp.
  2604--2620, Nov. 2020.

\bibitem{zhang2021intelligent}
S.~Zhang and R.~Zhang, ``Intelligent reflecting surface aided multi-user
  communication: Capacity region and deployment strategy,'' \emph{IEEE Trans.
  Commun}, 2021, Early Access.

\bibitem{you2020deploy}
C.~You and R.~Zhang, ``Wireless communication aided by intelligent reflecting surface: Active or passive?'' \emph{IEEE Wireless Commun. Lett.}, vol. 10, no. 12, pp. 2659-2663, Dec. 2021.

\bibitem{you2021enabling}
C.~You, Z.~Kang, Y.~Zeng, and R.~Zhang, ``Enabling smart reflection in integrated air-ground wireless network: IRS meets UAV,'' \emph{arXiv preprint arXiv:2103.07151}, 2021.

\bibitem{sen1}
R. S. Sankar, B. Deepak, S. P. Chepuri, ``Joint communication and radar sensing with reconfigurable intelligent surfaces,"
\emph{arXiv
  preprint arXiv:2105.01966}, 2021.

\bibitem{sen2}
W. Lu, Q. Lin, N. Song, Q. Fang, X. Hua, and B. Deng, ``Target detection in intelligent reflecting surface aided distributed MIMO radar systems," \emph{IEEE Sensors Lett.}, vol. 5, no. 3, pp. 1-4, Mar. 2021.

\bibitem{sen3}
H. Zhang, H. Zhang, B. Di, K. Bian, Z. Han, and L. Song, ``MetaLocalization: Reconfigurable intelligent surface aided multi-user wireless indoor localization," \emph{IEEE Trans. Wireless Commun.}, Jun. 2021, Early Access.

\bibitem{sen4}
J. Yao, Z. Zhang, X. Shao, C. Huang, C. Zhong, and X. Chen, ``Concentrative intelligent reflecting surface aided computational imaging via fast block sparse Bayesian learning," \emph{IEEE Veh. Techno. Conf. (VTC)}, Jun. 2021, pp. 1-6.

\bibitem{sen5}
J. Hu et al., ``Reconfigurable intelligent surface based RF sensing: Design, optimization, and implementation," \emph{IEEE J. Sel. Areas Commun.}, vol. 38, no. 11, pp. 2700-2716, Nov. 2020.

\bibitem{sen6}
C. You, B. Zheng, and R. Zhang, ``Fast beam training for IRS-assisted multiuser communications," \emph{IEEE Wireless Commun. Lett.}, vol. 9, no. 11, pp. 1845-1849, Nov. 2020.

\bibitem{semi1}
A. Taha, M. Alrabeiah, and A. Alkhateeb, ``Enabling large intelligent surfaces with compressive sensing and deep learning," \emph{IEEE Access}, vol. 9, pp. 44304-44321, 2021.

\bibitem{semi2}
G. C. Alexandropoulos and E. Vlachos, ``A hardware architecture for reconfigurable intelligent surfaces with minimal active elements for explicit channel estimation," in Proc. \emph{IEEE Inter. Conf. Acoustics, Speech Signal Process. (ICASSP)}, May 2020, pp. 9175-9179.

\bibitem{semi3}
Y. Lin, S. Jin, M. Matthaiou, and X. You, ``Tensor-based algebraic channel estimation for hybrid IRS-assisted MIMO-OFDM," \emph{IEEE Trans. Wireless Commun.}, vol. 20, no. 6, pp. 3770-3784, June 2021.

\bibitem{cano}
AK. Shauerman and AA. Shauerman, ``Spectral-based algorithms of direction-of-arrival estimation for adaptive digital antenna arrays," \emph{Intern. Conf. Seminar Micro/Nano. Ele. Devices}, Novosibirsk, Russia, 2010, pp. 251-255.

\bibitem{mle}
H. L. Van Trees, \emph{Optimum Array Processing}. New York, NY, USA: Wiley-Interscience, 2002.

\bibitem{musica}
P. Stoica and A. Nehorai, ``Music, maximum likelihood and
Cramer-Rao bound," \emph{IEEE Trans. Acoust., Speech Signal Process.}, vol. 37, pp. 720-741, May 1989.

\bibitem{esp}
R. Roy and T. Kailath, ``ESPRIT-estimation of signal parameters via rotational invariance techniques," \emph{IEEE Trans. Acoust. Speech Signal Process}. vol. 37, no. 7, pp. 984-995, 1989.

\bibitem{zhengIRS}
B. Zheng and R. Zhang, ``IRS meets relaying: Joint resource allocation and passive beamforming optimization," \emph{IEEE Wireless Commun. Lett.}, vol. 10, no. 9, pp. 2080-2084, Sept. 2021.

\bibitem{qingtoward}
Q. Wu and R. Zhang, ``Towards smart and reconfigurable environment: Intelligent reflecting surface aided wireless network,"
\emph{IEEE Commun. Mag.}, vol. 58, no. 1, pp. 106-112, Jan. 2020.

\bibitem{smooth}
T.-J. Shan, M. Wax, and T. Kailath, ``On spatial smoothing for directionof-arrival estimation of coherent signals," \emph{IEEE Trans. Acoust. Speech Signal Process.}, vol. 33, no. 4, pp. 806-811, Aug. 1985.


%\bibitem{semi4}
%Q. Wu, S. Zhang, B. Zheng, C. You, and R. Zhang, ``Intelligent reflecting surface aided wireless communications: A tutorial," \emph{IEEE Trans. Commun.}, vol. PP, no. 99, pp. 1-1,  2021.

\bibitem{near1}
C. Feng, H. Lu, Y. Zeng, S. Jin, and R. Zhang, ``Wireless communication with extremely large-scale intelligent reflecting surface,"
\emph{arXiv
  preprint arXiv:2106.06106}, 2021.
%  arXiv preprint arXiv:2106.06106, 2021.

\bibitem{near2}
H. Lu and Y. Zeng, ``Communicating with extremely large-scale array/surface: unified modelling and performance analysis,"
\emph{arXiv
  preprint arXiv:2104.13162}, 2021.
%  arXiv preprint arXiv:2104.13162, 2021.

\bibitem{fish}
S. M. Kay, \emph{Fundamentals of Statistical Signal Processing: Estimation
Theory}, Englewood Cliffs, NJ: Prentice-Hall, 1993.

\bibitem{bin}
D. Johnson and D. Dudgeon, ``\emph{Array Signal Processing: Concepts and Techniques}. Prentice-Hall, 1993.


\bibitem{MIMOradar}
M. Liu, G. Hu, J. Shi, and H. Zhou, ``DOA estimation method for multi-path targets based on TR MIMO radar," \emph{The J. Eng.}, vol. 2019, no. 2, pp. 461-465, Jan. 2019.

\bibitem{probing}
P. Stoica, J. Li, and Y. Xie, ``On probing signal design for MIMO radar," \emph{IEEE Trans. Signal Process.}, vol. 55, no. 8, pp. 4151-4161, Aug. 2007.

\bibitem{snr}
P. Stoica and G. Ganesan, ``Maximum-SNR spatialtemporal formatting designs for MIMO channels," \emph{IEEE Trans. Signal Process.}, vol. 50, no. 12, pp. 3036-3042, Dec. 2002.

\bibitem{succ1}
B. Yao, W. Wang, and Q. Yin, ``DOD and DOA estimation in bistatic non-uniform multiple-input multiple-output radar systems," \emph{IEEE Commun. Lett.}, vol. 16, no. 11, pp. 1796-1799, Nov. 2012.

\bibitem{succ2}
A. Govinda Raj and J. H. McClellan, ``Single snapshot super-resolution DOA estimation for arbitrary array geometries," \emph{IEEE Signal Process. Lett.}, vol. 26, no. 1, pp. 119-123, Jan. 2019.

\bibitem{suff}
D. Basu, ``On statistics independent of a complete sufficient statistic," Selected Works of Debabrata Basu. Springer, New York, NY, 2011: 61-64.
\end{thebibliography}
\end{document}